\documentclass[12pt,a4paper]{article}
\usepackage[utf8]{inputenc}
\usepackage[a4paper, total={6.5in, 10in}]{geometry}
\usepackage{booktabs}
\usepackage{threeparttable}
\usepackage{tabularx}
\usepackage{rotating}
\usepackage{tikz}
\usetikzlibrary{positioning}
\usepackage{placeins}
\usepackage{chngcntr}
\usepackage{ragged2e}
\usepackage{authblk}
\usepackage{hyperref}
\usepackage{subcaption}
\usepackage{float}
\usepackage{xspace}
\hypersetup{
  colorlinks   = true,
  urlcolor     = blue,
  linkcolor    = blue,
  anchorcolor  = blue,
  citecolor    = blue
}
\usepackage{amsmath, amssymb}
\newtheorem{theorem}{Theorem}

\newtheorem{proposition}[theorem]{Proposition}

\newenvironment{proof}[1][Proof]{\noindent\textbf{#1.} }{\ \rule{0.5em}{0.5em}}
\providecommand{\U}[1]{\protect\rule{.1in}{.1in}}

\begin{document}

\title{Veblen effects and broken windows in an environmental OLG model\thanks{An earlier version of this paper was presented at the 10th Nonlinear Economic Dynamics Conference, Naples, Italy; at the 49th AMASES Annual Conference, Florence, Italy; at the Workshop Evolutionary Dynamics, Game Theory and Uncertainty in Economic Modelling, Milan, Italy; at the Systemic Risk Institute Workshop, Prague, Czechia; and at the CONTROPT Workshop 2025, Bari, Italy. We thank the participants for their insightful comments and suggestions. Nicolas Blampied acknowledges the support from the NPO ``Systemic Risk Institute'',  Czechia, LX22NPO5101, funded by the European Union Next Generation EU (Ministry of Education, Youth and Sports, NPO: EXCELES). Usual caveats apply.}
}

\author{Nicolás Blampied$^{1}$ \: $\cdot$ \: Alessia Cafferata$^{2}$ \\ Marwil J. Davila-Fernandez$^{3}$}
\date{
{\normalsize {$^{1}$\emph{Department of Economics, Masaryk University}}}\\
{\normalsize {$^{2}$\emph{Department of Economics and Management, University of Pisa}}}\\
{\normalsize {$^{3}$\emph{Department of Economics, Colorado State University\smallskip}}\smallskip}\\
\bigskip December 2025}
\maketitle

\begin{abstract}
Can constantly comparing ourselves to others lead to overconsumption, ultimately increasing the ecological footprint? How do social comparisons shape green preferences over time? To answer these questions, we develop an environmental Overlapping Generations (OLG) model that explicitly accounts for Veblen effects and allows green preferences to be updated asynchronously, influenced by past environmental conditions and relative status considerations. We show that, along the optimal path, positional spending leads to overconsumption, which is detrimental to the environment. Taxing consumption is counterproductive as it does not directly address the social comparisons issue, leaving the problem unchanged. When the Veblenian mechanism is weak, the introduction of a materialistic ``secular trend'' -- that lowers the importance placed on the public good -- gives rise to two stable equilibria separated by a saddle: one in which agents care about environmental quality as much as consuming, and the other in which they derive utility solely from the latter. Studying the basins of attraction reveals that green investments are highly fragile. Our numerical experiments further indicated that, when Veblen effects are strong, the model depicts endogenous, persistent, aperiodic oscillations. In this case, green preferences fluctuate close to zero, and environmental quality is very low. Taken together, these findings suggest environmental vulnerability grows in parallel with status-driven consumption.

\bigskip

 \bigskip\ 

\end{abstract}

\newpage

\section{Introduction}

Veblen effects describe consumption driven by considerations of relative status. People value consuming both for its use and for the social position that it represents. The latter fuels a sort of arms race, where the living standards of wealthier reference groups become unattainable targets for the rest. Veblen referred to such a process as pecuniary emulation (\href{#Veblen 1899}{Veblen, 1992 [1899]}; see also \href{#Bagwell and Bernheim 1996}{Bagwell and Bernheim, 1996}; \href{#Bowles and Park 2005}{Bowles and Park, 2005}; \href{#Behringer et al 2024}{Behringer et al., 2024}). With inequality on the rise in both developed and developing countries, social media has made the lifestyles of the wealthy more visible than ever. This increased exposure arguably further encourages conspicuous spending, raising questions about environmental sustainability. Can constantly comparing ourselves to others lead to overconsumption and ultimately result in higher ecological footprints? How do social comparisons influence green preferences over time?

To answer these questions, we develop an environmental Overlapping Generations (OLG) model that is novel in explicitly accounting for Veblen effects. As in \href{#Bowles and Park 2005}{Bowles and Park (2005)}, agents are assumed to derive utility from ``effective'' consumption, defined as the difference between an individual's actual consumption and what the reference group s/he aspire to be spends. In addition, and similar to \href{#John and Pecchenino 1994}{John and Pecchenino (1994)}, we assume that people also value environmental quality, while their income can either be consumed and/or allocated to conservation activities. Still, we innovate by allowing green preferences to be updated asynchronously, influenced by past environmental conditions and relative status considerations (see \href{#Hommes et al 2005}{Hommes et al., 2005}). This mechanism is motivated by a ``broken windows'' effect, where agents value the environment less when they do not feel the benefits of investing in it.

We show that, along the optimal path, status-driven consumption leads to overconsumption. The higher the reference group spends, the higher individual consumption will be. Given the budget constraint, this comes at the expense of resources devoted to environmental quality. We additionally document that taxing consumption could be counterproductive as it does not directly address the social comparisons problem. When the Veblenian mechanism is weak, the introduction of a materialistic ``secular trend'' into the updating green preferences function, which exogenously reduces the value agents place on the public good, makes the model compatible with multiple equilibria, though in a non-monotonic way. A strong trend leads to a unique equilibrium in which agents do not value the environment. Still, for intermediary values of this parameter, the model features two extreme equilibria. One in which agents care about environmental quality nearly as much as consumption, and another in which they derive utility only from the latter. These are separated by a saddle point whose separatrix determines the attracting regions of the stable solutions.

Studying the basins of attraction in the multiple equilibria case reveals several disconnected regions, a result that follows from the map being non-invertible. From an economic point of view, this property highlights the sensitivity and complexity of green preferences. A small negative shock near the equilibrium where agents value the environment highly can push the system into the basin of the ``bad'' point, while a small positive shock near the ``bad'' solution can make the economy jump towards the ``good'' one. Our numerical experiments further indicated that the model is compatible with endogenous, persistent, aperiodic oscillations. It is shown that the combination of strong Veblen and broken windows effects is associated with a Neimark-Sacker bifurcation. In this scenario, green preferences fluctuate close to zero, and environmental quality is very low. Such a result is robust to variations in the degree of inertia of green preferences. We document that as we reduce inertia, a Flip bifurcation might also occur. Taken together, these findings suggest environmental vulnerability grows in parallel with status-driven consumption.


This paper contributes to the literature on environmental OLG models, exploring a novel channel that dialogues with discrete choice theory (\href{#Brock and Hommes 1998}{Brock and Hommes 1998}; \href{#Brock and Durlauf 2001}{Brock and Durlauf, 2001}). OLGs stand as a popular tool among scholars interested in green-growth dynamics (e.g. \href{#John and Pecchenino 1994}{John and Pecchenino, 1994}; \href{#Mariani et al 2010}{Mariani et al., 2010}; \href{#de la Croix and Gosseries 2012}{de la Croix and Gosseries, 2012}; \href{#Wei and Aadland 2022}{Wei and Aadland, 2022}; \href{#Jaimes 2023}{Jaimes, 2023}), providing a tractable structure that allows to study situations where agents' actions have consequences that outlive them. Recent examples include fiscal and monetary policy considerations (see \href{#Davila-Fernandez et al 2025}{Davila-Fernandez et al., 2025}; \href{#Grassetti et al 2025}{Grassetti et al., 2025}). Coupled with an evolutionary module, they have proven to be a useful tool to investigate more complex dynamics in nonlinear systems (as in \href{#Zhang 1999}{Zhang, 1999}; \href{#Antoci et al 2016}{Antoci et al., 2016}; \href{#Antoci et al 2019}{2019}; \href{#Caravaggio and Sodini 2023}{Caravaggio and Sodini, 2023}). While discrete-choice theory with social interactions has been used to study environmental problems (e.g. \href{#Zeppini 2015}{Zeppini 2015}; \href{#Cafferata et al 2021}{Cafferata et al., 2021}; \href{#Annicchiarico et al 2024}{Annicchiarico et al., 2024}; \href{#Campiglio et al 2024}{Campiglio et al., 2024}), to the best of our knowledge, we are the first to adopt it in an OLG setup.

The remainder of the paper is organised as follows. Section 2 provides an overview of stylised facts linking the environment to consumption in OECD countries. In Section 3, we present our OLG model, derive the conditions for the existence of a unique or multiple equilibria analytically, and perform the corresponding local stability analysis. We also discuss the implications of adopting a consumption tax on our dynamic system. Section 4 reports a series of numerical experiments that bring concreteness to our main findings. Some final considerations follow.

\section{Inequality, consumption, and pollution}

Given the role of social norms in determining consumption patterns and their potential negative environmental implications, it is useful to look at scatterplots to gain some initial insights into what is going on. Because Veblen effects are driven by comparisons with high-income groups, we employ an inequality measure that is especially sensitive to the upper tail of the distribution. Using data from OECD countries, Fig. \ref{fig:StylizedFacts}, panel (a), illustrates the positive correlation between household consumption and wealth inequality. It also shows, on panel (b), that carbon emissions per capita are positively associated with the share of consumption in GDP.\footnote{According to \href{https://www.oecd.org/en/data/insights/statistical-releases/2024/05/share-of-oecd-economies-in-global-gdp-broadly-stable-at-46-in-2021-compared-to-2017.html}{OECD (2024)}, member countries represented about 46\% of global GDP in 2021. With the emergence of Brazil, Russia, India, and China (BRICs) as major players in the global economy, this share has been reduced. Still, the average weight for the period 2010-2023 is higher, with the OECD accounting for around 60\% of world GDP and 50\% of CO\textsubscript{2} emissions.} Such a relationship is consistent with the idea that increasing inequality drives status-driven consumption through direct social comparisons with richer reference groups. This mechanism is particularly relevant considering that between 1990 and 2020, two-thirds (one-fifth) of global warming is attributable to the wealthiest 10\% (1\%), with individual contributions 6.5 (20) times the average per capita share (\href{#Schongart et al 2025}{Schongart et al., 2025}). 

\begin{figure}[tbp]
  \centering
  \begin{minipage}{3in}
    \centering
    \includegraphics[width=\linewidth]{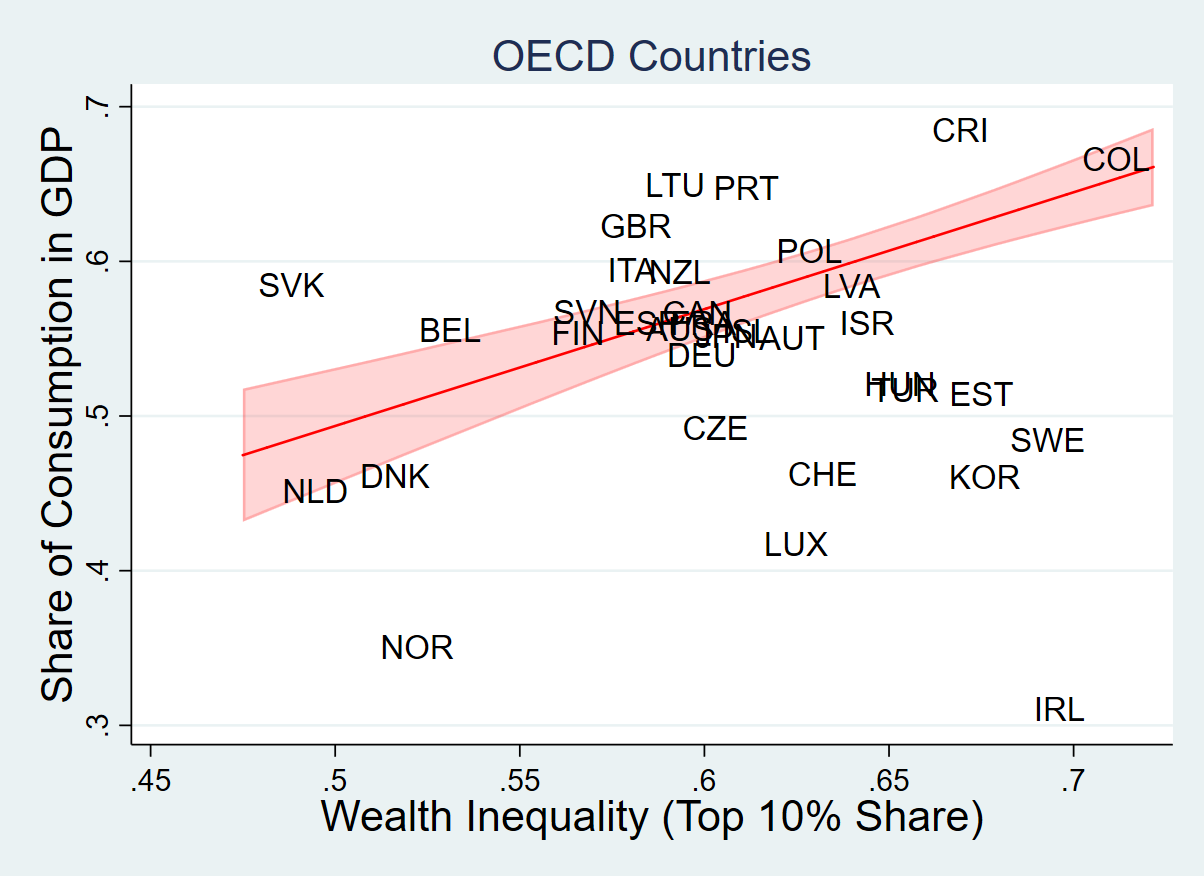}
    \\[-1ex]
    {\footnotesize (a)}
  \end{minipage}%
  \hspace{10mm}%
  \begin{minipage}{3in}
    \centering
    \includegraphics[width=\linewidth]{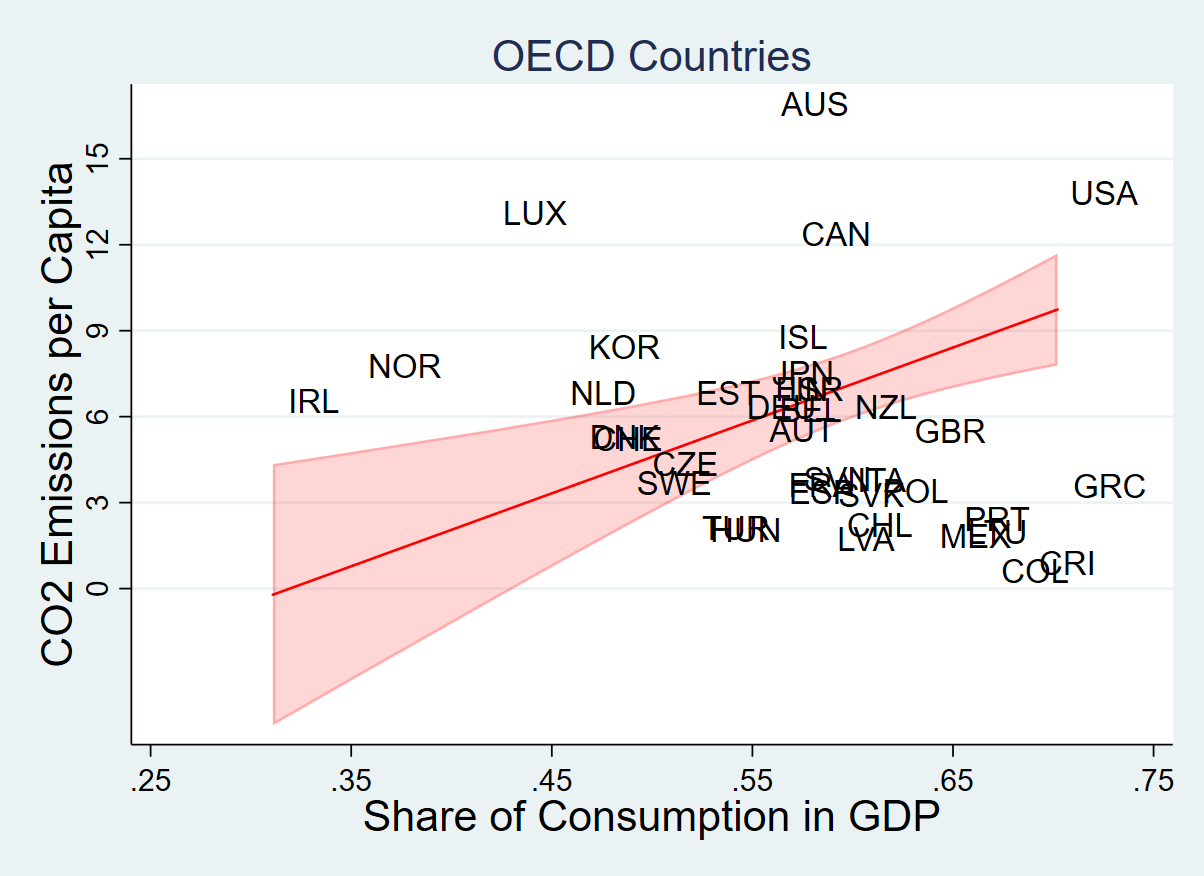}
    \\[-1ex]
    {\footnotesize (b)}
  \end{minipage}
  \caption{Share in Consumption, Wealth Inequality and Carbon Intensity.}
  \label{fig:StylizedFacts}
  \vspace{0.2cm}
  \begin{minipage}{1\textwidth}
    \footnotesize
    Notes: Panel (a) depicts the association between the share of consumption in GDP and wealth inequality, while Panel (b) shows the relationship between carbon emissions per capita and the share of consumption in GDP. Carbon emissions per capita are measured in tons. Both relationships are estimated using population weights. The figure was elaborated by the authors based on data from the World Bank and the World Inequality Database (WID), covering the period 2010–2023.
  \end{minipage}
\end{figure}

The Veblen story, in itself, does not imply the presence of a “broken windows” effect. The latter is related to the idea that initial signs of antisocial behaviour and civil disorder, small but untreated, might lead citizens to withdraw from and/or value the community less, ultimately resulting in higher crime. In the context of environmental sustainability, we reinterpret the effect: the lower the environment's quality, the less agents value it as a public good. Using OECD data, we turn to the relationship between carbon emissions and wealth inequality. Fig. \ref{fig:StylizedFacts1}, panel (a), reports that the top 10\% wealth-share is positively related to per capita carbon emissions. Still, such a relationship does not hold among low-emission countries, as indicated in panel (b).

\begin{figure}[tbp]
  \centering
  \begin{minipage}{3in}
    \centering
    \includegraphics[width=\linewidth]{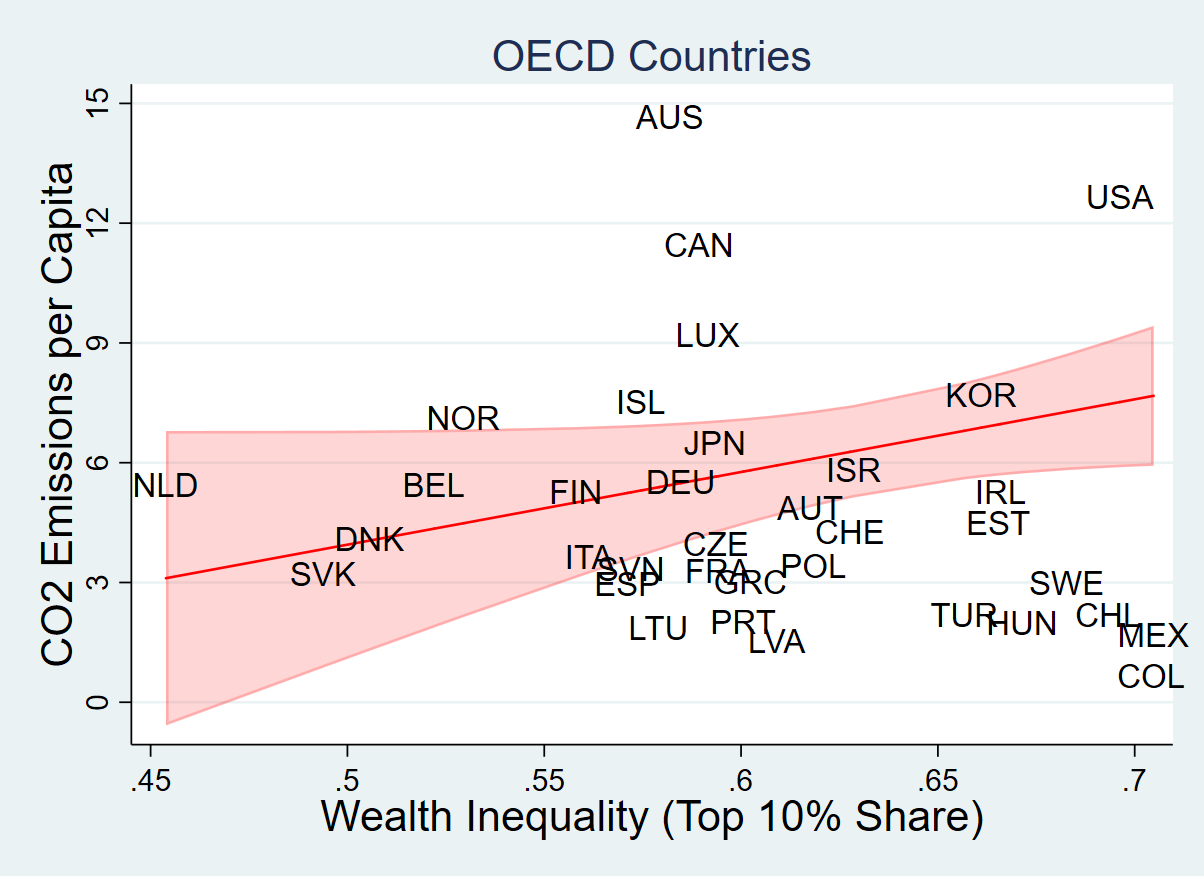}
    \\[-1ex]
    {\footnotesize (a)}
  \end{minipage}
  \hspace{10mm}
  \begin{minipage}{3in}
    \centering
    \includegraphics[width=\linewidth]{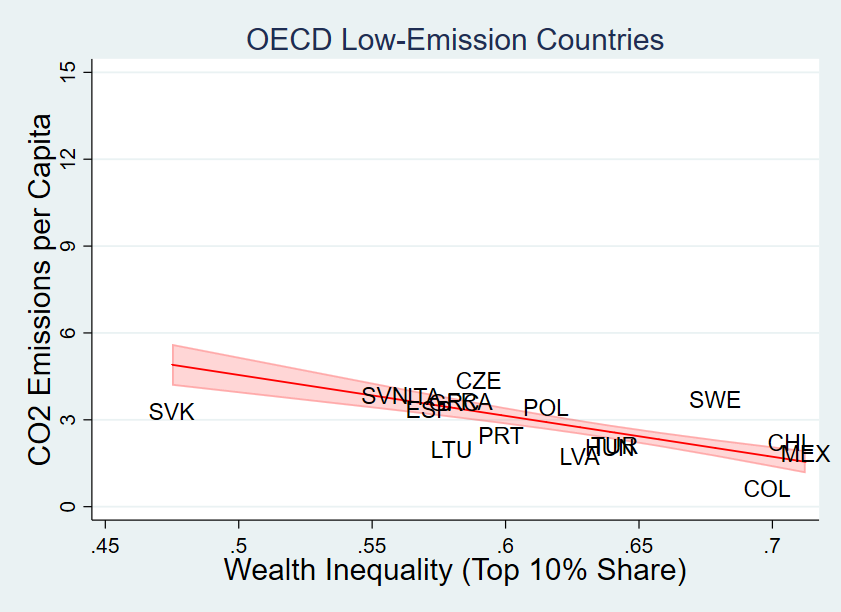}
    \\[-1ex]
    {\footnotesize (b)}
  \end{minipage}
  \caption{Wealth Inequality and Carbon Intensity.}
  \label{fig:StylizedFacts1}
  \vspace{0.2cm}
  \begin{minipage}{1\textwidth}
    \footnotesize
    Notes: Low–emission countries are defined as those OECD countries below the median level of carbon emissions per capita. For further details, see Figure~\ref{fig:StylizedFacts}.
  \end{minipage}
\end{figure}

\begin{figure}[tbp]
  \centering
  \begin{minipage}{3in}
    \centering
    \includegraphics[width=\linewidth]{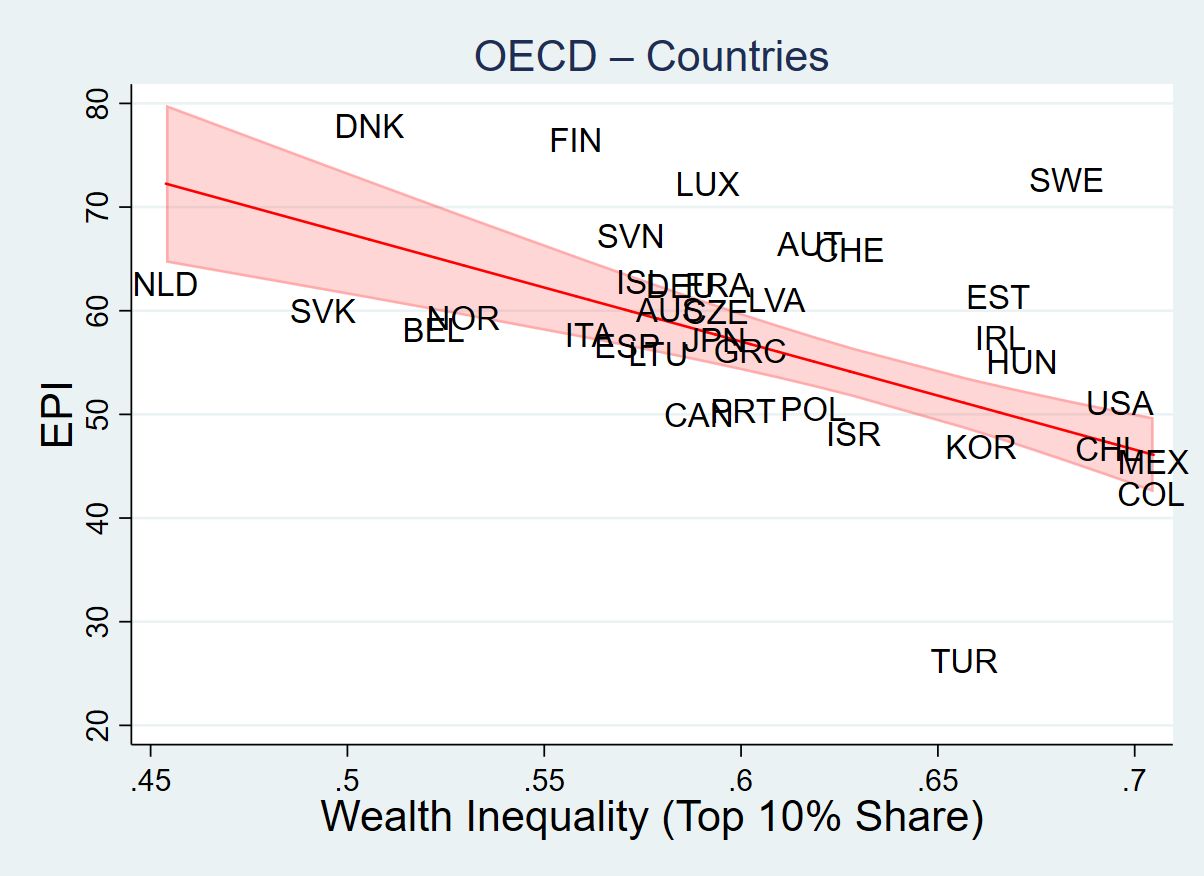}
    \\[-1ex] 
    {\footnotesize (a)}
  \end{minipage}%
  \hspace{10mm}%
  \begin{minipage}{3in}
    \centering
    \includegraphics[width=\linewidth]{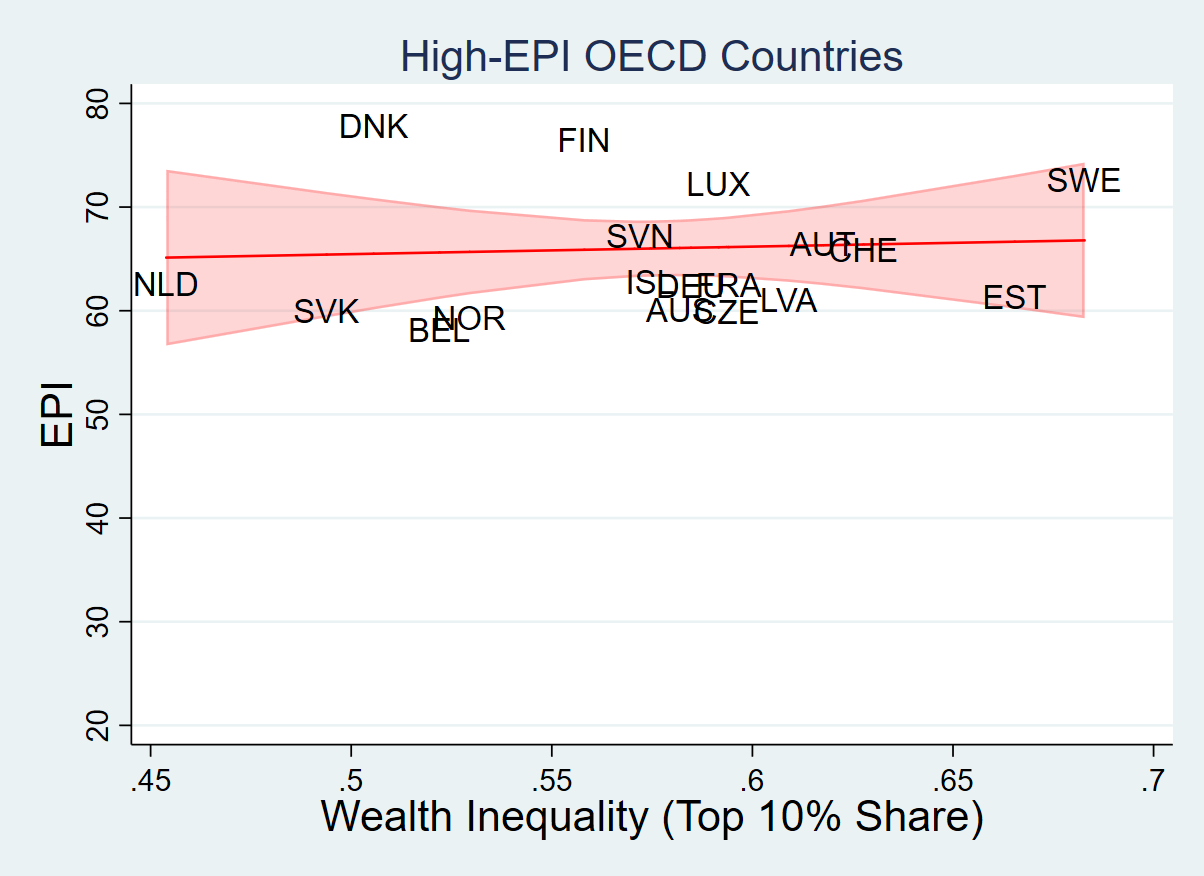}
    \\[-1ex]  
    {\footnotesize (b)}
  \end{minipage}
  \caption{Wealth Inequality and Environmental Performance Index (EPI).}
  \label{fig:StylizedFacts2}
  \vspace{0.2cm}  
  \begin{minipage}{1\textwidth}
    \footnotesize
    Notes: High EPI countries are defined as OECD members with scores above the median. Unlike carbon emissions per capita, higher scores indicate better environmental performance. The EPI was retrieved from the Yale Center for Environmental Law \& Policy, and the data corresponds to the year 2022. 
  \end{minipage}
\end{figure}

If instead we turn to a more comprehensive indicator of environmental quality, we obtain an equivalent correspondence. Fig. \ref{fig:StylizedFacts2}, panel (a), shows a negative relationship between wealth concentration and the Environmental Performance Index (EPI). The EPI is a composite indicator developed by the \textit{Yale Center for Environmental Law \& Policy} and the \textit{Center for International Earth Science Information Network Earth Institute}. It is a national-level measure of how close countries are to the established environmental policy objectives. The EPI ranks up to 180 countries regarding climate change performance, environmental health, and ecosystem vitality. Its range goes from 0 to 100, the highest possible performance. As in the previous case, nonetheless, for high EPI countries, the relationship with the wealth share of the top 10\% disappears.


The change in slope between the two panels in Figs. \ref{fig:StylizedFacts1} and \ref{fig:StylizedFacts2} suggests that the preferences might not be homogeneous. We argue that the positive (negative) relationship between inequality and emissions (or the EPI) that breaks down in low-polluting countries hints at a “broken windows” effect. If so, green preferences should be taken as endogenous to the current state of the environment. Agents care more about this public good when they perceive that others are also doing their part. Environmental preferences are likely to interact with the social norms influencing consumption. We proceed by developing an OLG model that provides a simple yet coherent narrative of the involved channels. 

\section{The model}

We consider an OLG economy where agents live for two periods. They work during the first phase, receiving an exogenously determined income, and choose how much to consume. Savings are then dedicated to improving a public good, the environment, from which they get utility in the second phase. Our study builds on \href{#Bowles and Park 2005}{Bowles and Park (2005)} treatment of Veblen effects in the lifetime choice optimisation problem. Moreover, our environmental module dwells on \href{#John and Pecchenino 1994}{John and Pecchenino (1994)}, which offers a simple and tractable framework that is quite popular in the green-growth literature  (e.g. \href{#Wei and Aadland 2022}{Wei and Aadland, 2022}; \href{#Caravaggio and Sodini 2023}{Caravaggio and Sodini, 2023}; \href{#Jaimes 2023}{Jaimes, 2023}; \href{#Davila-Fernandez et al 2025}{Davila-Fernandez et al., 2025}). Environmental conditions evolve in such a way that consumption has a damaging effect, while conservation investments improve them.

Besides introducing Veblen effects into this family of models for the first time, our study also comes with a second twist. Green preferences are not assumed to be constant but are updated through a discrete choice mechanism (see \href{#Brock and Hommes 1998}{Brock and Hommes, 1998}; \href{#Brock and Durlauf 2001}{Brock and Durlauf, 2001}; \href{#Hommes et al 2005}{Hommes et al., 2005}). We motivate it as a ``broken windows'' effect. The utility an agent derives from the public good increases with its quality. Agents value the environment less when they do not feel the benefits of investing in it. To focus on the interaction between these two mechanisms, we abstract from physical or human capital considerations. This simplification allows us to explore the model's dynamic properties more deeply without loss of generality, as incorporating an explicit damage function in the production technology would actually accentuate the effects from the transmission channels already accounted for. Future research should extend the present framework to a growth OLG set-up. Fig. \ref{Diagram} provides a summarising diagram of the model.

\begin{figure}[tbp]
  \centering
  \includegraphics[width=3.25in]{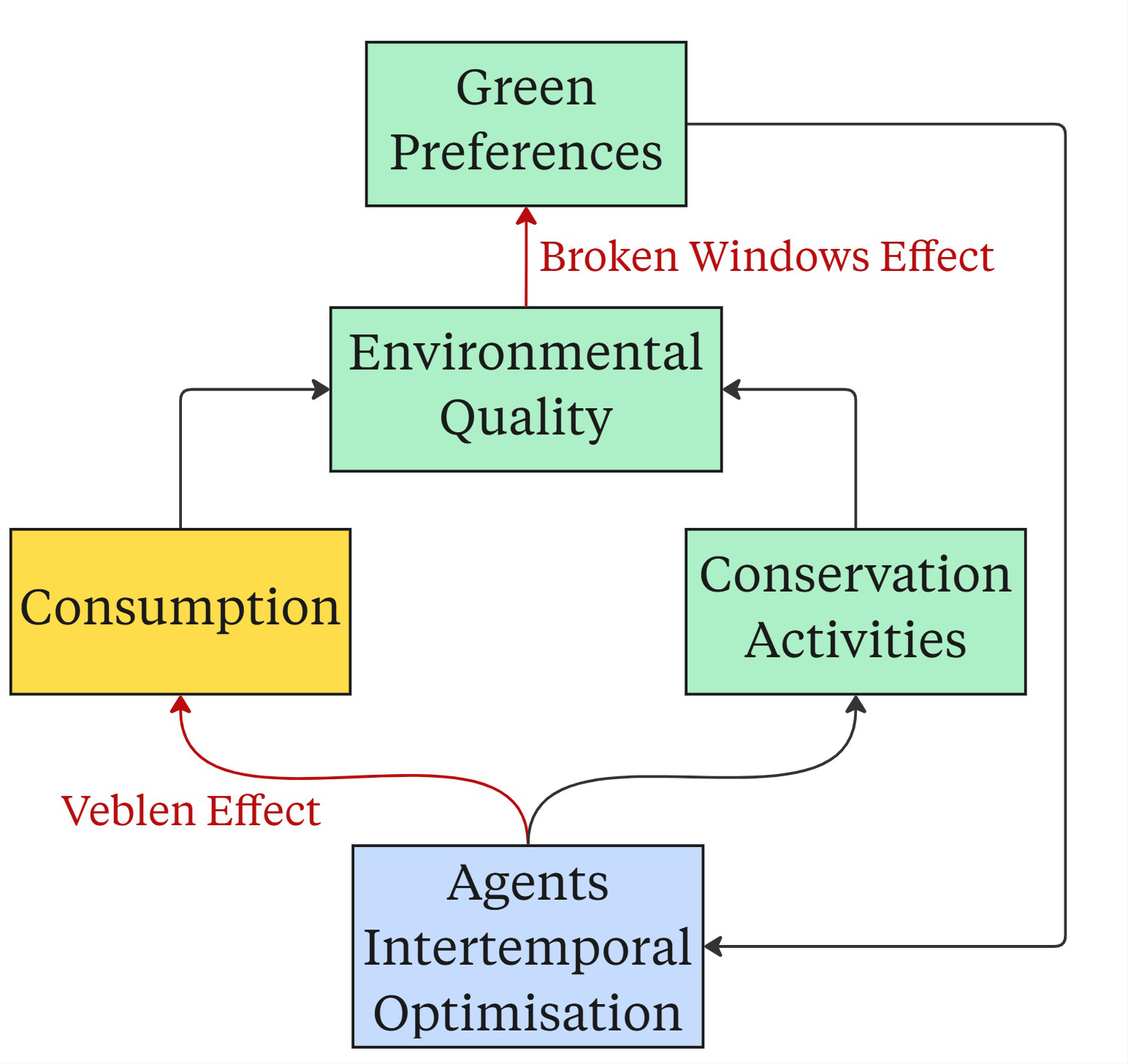}
  \caption{Wealth inequality and carbon intensity.}\label{Diagram}
\end{figure}

\subsection{Veblen, preferences, and the environment}

Suppose agents get utility from \textquotedblleft effective\textquotedblright%
\ consumption ($c^{\ast}$) when they are young, reflecting their immediate need to accumulate goods early in life, and environmental quality ($e$)
at the second stage of their lives, as the marginal benefit of additional consumption diminishes. A simple specification of the corresponding utility function ($U$) is:%
\[
U=\ln c_{t}^{\ast}+\pi_{t}\ln e_{t+1}%
\]
where $\pi$ captures how much they value environmental conditions, i.e. green preferences.

Effective consumption embodies a critical proposition underlying Veblen's account. People compare consumption and choose their spending to be more like a higher reference group (\href{#Bowles and Park 2005}{Bowles and Park, 2005}). It is defined as the difference between how much you actually consume ($c$) and how much the reference group you aspire to be spends ($\tilde{c}$),
so that we write:
\begin{equation}
c_{t}^{\ast}=c_{t}-\upsilon\tilde{c}\label{Veblen}%
\end{equation}
where $\upsilon\geqslant0$ is a parameter capturing how much agents care about
the others' consumption. For any value of $\upsilon>0$, there is some degree of social comparisons that embodies the idea of \textquotedblleft Keeping up with the Joneses\textquotedblright.

We assume the adult population is self-employed, so output equals their income. There is no capital in the economy, and agents receive wages ($w$) that can be either consumed or used in conservation activities ($m$). Therefore, the budget constraint is such that:%
\[
w=c_{t}+m_{t}%
\]

Consumption harms the environment by depleting natural resources and
generating different types of pollution. Still, conservation efforts work in
the opposite direction, contributing to the recovery of biodiversity. As a
simple representation of the interaction between these two forces, we refer to
\href{#John and Pecchenino 1994}{John and Pecchenino (1994)} and write:%
\begin{equation}
e_{t+1}= e_{t}+\sigma m_{t} -\gamma c_{t}\label{JP1994}%
\end{equation}
where $\sigma>0$ and $\gamma>0$ are technical parameters capturing the respective marginal effects of $m$ and $c$ on our environmental quality.

The agent's life-choice problem consists of maximising utility given the budget and environmental constraints:
\begin{align*}
& \underset{c_{t}~m_{t}}{\max}\ln c_{t}^{\ast}+\pi_{t}\ln e_{t+1}\\
& \text{s.t.}\\
w  & =c_{t}+m_{t}\\
e_{t+1}  & = e_{t}-\gamma c_{t}+\sigma m_{t}%
\end{align*}
Substituting both restrictions into the objective function and making use of Eq. (\ref{Veblen}), the optimisation process simplifies to:
\[
\underset{c_{t}}{\max}\ln\left(  c_{t}-\upsilon\tilde{c}\right)  +\pi_{t}%
\ln\left[e_t + \sigma\left(
w-c_{t}\right) -\gamma c_{t}  \right]
\]

The First Order Conditions (FOC) indicate that, in the optimal path, Veblen effects are associated with overconsumption. Consequently, they also reduce resources allocated to conservation activities:
\begin{align}
c_{t}& =\frac{\sigma w+e_{t}}{\left( 1+\pi _{t}\right) \left( \gamma +\sigma
\right) }+\underset{\text{Overconsumption}}{\underbrace{\left( \frac{\pi _{t}%
}{1+\pi _{t}}\right) \upsilon \tilde{c}}}  \nonumber \\
&  \label{FOC} \\
m_{t}& =\frac{\gamma w-e_{t}}{\left( 1+\pi _{t}\right) \left( \gamma +\sigma
\right) }+\left( \frac{\pi _{t}}{1+\pi _{t}}\right) \left( w-\upsilon \tilde{%
c}\right)   \nonumber
\end{align}
Green preferences mediate such trajectories. In the absence of social comparisons, from the first expression in (\ref{FOC}), a higher $\pi$ is related to lower consumption. This result is intuitive, as consumption at $t$ negatively impacts environmental quality in $t+1$. Thus, an agent that values the environment more will optimally consume less. What is less intuitive is that a higher $\pi$ is also associated with lower $m$. The explanation is that if an agent values the environment highly, the level of environmental quality required to achieve the same level of utility will be lower. Therefore, s/he will invest less in conservation activities.

Status-driven considerations complicate things a little. For moderate values of $v \tilde{c}$, we still have a negative relationship between consumption and green preferences. However, the more households care about the reference group ($v$), and/or the higher the consumption of those you aspire to be ($\tilde{c}$), the more overconsumption is fueled by an increasing $\pi$. This counterintuitive finding is a consequence of the interaction between social comparisons with the utility agents derive from the environment. As environmental quality deteriorates due to status-driven consumption, individuals derive less utility from the public good and are pushed to rely even more heavily on positional consumption. The magnitude of this compensatory behaviour is proportional to how much agents value environmental quality in their utility functions. For similar reasons, the higher $v \tilde{c}$ is, stronger green preferences paradoxically reduce conservation investments.

Substituting (\ref{FOC}) into Eq. (\ref{JP1994}), we obtain the dynamics of
environmental quality as a function of agents' level of green preferences and Veblen
effects:%
\[
e_{t+1}=\left(  \frac{\pi_{t}}{1+\pi_{t}}\right)  \left[ e_{t}+\sigma
w-\left(  \sigma+\gamma\right)  \upsilon\tilde{c}\right]
\]

\subsection{Asynchronous updating of beliefs}

While in a more conventional set-up, preferences are assumed to be fixed and exogenously given, we instead take an evolutionary stance on the issue and allow them to vary. We describe the dynamics of green preferences by referring to the discrete choice literature (see \href{#Brock and Hommes 1998}{Brock and Hommes 1998}; \href{#Brock and Durlauf 2001}{Brock and Durlauf, 2001}). More specifically, we rely on specifications used when there is asynchronous updating of beliefs (as in \href{#Hommes et al 2005}{Hommes et al., 2005}). Agents face a probability of valuing the environment ($p$), so
that:
\begin{equation}
\pi_{t+1}=\alpha\pi_{t}+\left(  1-\alpha\right)  p_{t}\label{Altruism}%
\end{equation}
where $0<\alpha<1$ captures the degree of inertial of the behaviour or trait.

To formalise how such a probability changes over time, we refer to the ``broken windows'' metaphor. Introduced by \href{#Wilson and Kelling 1982}{Wilson and Kelling's (1982)} short article published in \textit{The Atlantic}, broken windows stand as a sign of anti-social behaviour and civil disorder, initially small and that goes untreated, but that leads citizens to become fearful and withdraw from the community, ultimately leading to higher disorder and crime. The metaphor links the occurrence of serious crimes with visible signs of incivility in a community, having gained its space in the economic-criminology literature (e.g. \href{#Corman and Mocan 2002}{Corman and Mocan, 2002}; \href{#van der Weele et al 2021}{van der Weele et al., 2021}; \href{#Miceli and Segerson 2024}{Miceli and Segerson, 2024}). We import this concept and apply it to how agents face environmental conditions: they value the environment less when they do not feel the benefits of investing in it. In mathematical terms:%
\begin{equation}
p_{t}=\frac{1}{1+\exp\left(  -\beta e_{t}+ \rho\right)  }\label{Probability}%
\end{equation}
where $\beta$ is the intensity of choice parameter (for a literature review, see \href{#Franke and Westerhoff 2017}{Franke and Westerhoff, 2017}) and $\rho \geq 0$ represents a materialistic ``secular trend'' that exogenously reduces how much agents value the environment.

Substituting Eq. (\ref{Probability}) into (\ref{Altruism}), we obtain:%
\begin{equation*}
\pi_{t+1}=\alpha\pi_{t}+\frac{1-\alpha}{1+\exp\left(  -\beta e_{t}+ \rho  \right)  }\label{AltruismVeblen}%
\end{equation*}

\subsection{Dynamic system}

Our two-dimensional map is given by:
\begin{equation}
M_{1}:=\left\{
\begin{array}
[c]{l}%
e_{t+1}=\left(  \dfrac{\pi_{t}}{1+\pi_{t}}\right)  \left[  e_{t}+\sigma
w-\left(  \sigma+\gamma\right)  \upsilon\tilde{c}\right]  \\
\pi_{t+1}=\alpha\pi_{t}+\dfrac{1-\alpha}{1+\exp\left(  -\beta e_{t}+ \rho  \right)  }%
\end{array}
\right.  \label{M1_2D}%
\end{equation}
In steady-state, $e_{t}=e_{t+1}=\bar{e}$ and $\pi_{t}=\pi_{t+1}=\bar{\pi}$. Therefore, the equilibrium conditions imply:%
\begin{align}
\bar{e}  & =\bar{\pi}\left[  \sigma w-\left(  \sigma+\gamma\right)
\upsilon\tilde{c}\right] \nonumber \\ \label{eq_cond} \\
\bar{\pi}  & =\frac{1}{1+\exp\left(  -\beta \bar{e}+ \rho  \right)
} \nonumber
\end{align}
The first expression in (\ref{eq_cond}) establishes a linear relationship between environmental quality and green preferences that crosses the origin. It will be positive only if Veblen effects are weak. The second expression shows instead $\bar{\pi}$ as an S-shaped function of $\bar{e}$. The materialistic ``secular trend'' moves the S-curve along the horizontal axis, with higher values of $\rho$ shifting it to the right.

We are ready to state and prove the following Propositions regarding the existence of a unique or multiple equilibria, and their local stability properties.

\begin{samepage}
\begin{proposition}
    When Veblen effects are sufficiently strong, such that:
    \begin{equation*}
        \frac{v \tilde{c}}{w}>\frac{\sigma}{\sigma + \gamma}
    \end{equation*}
    The dynamic system $M_1$ admits a unique equilibrium solution with $\bar{e}\leq 0$. On the other hand, when status-driven consumption is moderate or weak:
        \begin{equation*}
        \frac{v \tilde{c}}{w}<\frac{\sigma}{\sigma + \gamma}
    \end{equation*}
    Then, $M_1$ admits up to three equilibria. The cardinality of the solution set is non-monotonic and exhibits a hump-shaped dependence on $\rho$.
\end{proposition}

\begin{proof}
    See Appendix A.1.
\end{proof}
\end{samepage}

\medskip



\begin{samepage}
\begin{proposition} \label{MainProp}
    An equilibrium point of $M_1$ is locally stable provided that:
\begin{align*}
    \eta & <1 - \alpha \\ \eta & < \left(  1+\alpha\right) \left(1+2\bar{\pi}\right)  \\
    -\frac{1+\eta}{\bar{\pi}} &<1-\alpha
\end{align*}
where $\eta=\left(\bar{e}/\bar{\pi}  \right)\left(\partial \pi_{t+1}/ \partial e_t \right)$ is the elasticity of green preferences with respect to environmental conditions. A violation of the first inequality while the other two are satisfied is associated with a Fold bifurcation. When only the second condition is violated, a Flip bifurcation occurs. Finally, if all conditions except the last are satisfied, we have a Neimark-Sacker bifurcation.
\end{proposition}

\begin{proof}
    See Appendix A.2.
\end{proof}
\end{samepage}

\bigskip

Together, these Propositions highlight some major consequences of Veblen effects on environmental quality. When people care too much about reference-group consumption ($v$) and/or the rich spend excessively ($\tilde{c}$), the system becomes more likely to settle into a unique equilibrium in which agents place no value on the environment simply because it is so degraded that they no longer perceive any worth in it. Additionally, because $\eta$ depends on $v\tilde{c}$, status-driven spending increases the likelihood of persistent endogenous oscillations. A more moderate Veblenian channel, by contrast, allows the map to admit a more desirable equilibrium in which agents value private consumption and the public good more evenly. This last outcome still depends, however, on the strength of the underlying materialist trend.

\subsection{Taxing consumption}

In an attempt to address the problem of overconsumption and its negative
ecological implications, some scholars have proposed the creation of a
consumption tax ($\tau$).\footnote{An intuitive example with a strong Veblenian component refers to the fast-fashion sector. It specialises in the production of clothing that often mimics popular styles from fashion labels, big-name brands, and independent designers. The ``eco-tax'' on fast-fashion aims to curb environmentally harmful consumption, and has been recently adopted (or proposed) in various European contexts. For instance, the law passed by the French Parliament in 2025 imposed an additional tax per item on ultra-fast fashion clothing (see also Niinimäki et al., 2020).} We turn our attention to this proposal in the context of the present model. The intertemporal optimisation problem becomes:
\begin{align*}
& \underset{c_{t},m_{t}}{\max}\ln c_{t}^{\ast}+\pi_{t}\ln e_{t+1}\\
& \text{s.t.}\\
w  & =\textcolor{red}{\left(1+\tau\right)}c_{t}+m_{t}\\
e_{t+1}  & =e_{t}-\gamma c_{t}+\sigma m_{t}%
\end{align*}
which can be simplified and rewritten as:%
\[
\underset{c_{t}}{\max}\ln\left(  c_{t}-v\tilde{c}\right)  +\pi_{t}\ln\left[
e_{t}-\gamma c_{t}+\sigma \left(  w- \textcolor{red}{\left(1+\tau\right)}  c_{t}\right)
\right]
\]

The FOCs now imply:
\begin{align*}
c_{t}  & =\frac{\sigma w+e_{t}}{\left(  1+\pi_{t}\right)  \left[
\gamma+\sigma\textcolor{red}{\left(1+\tau\right)}  \right]  }+\underset
{\text{{\footnotesize Overconsumption}}}{\underbrace{\left(  \frac{\pi_{t}%
}{1+\pi_{t}}\right)  v\tilde{c}}}\\
m_{t}  & =\frac{\gamma w-\textcolor{red}{(1+\tau)}e_{t}}{\left(  1+\pi_{t}\right)  \left[
\gamma+\sigma\textcolor{red}{\left(1+\tau\right)}  \right]  }+\left(  \frac{\pi_{t}}%
{1+\pi_{t}}\right) \left(  w-v\tilde{c}\right)
\end{align*}
Because $c$ and $\tau$ are inversely related, it does seem that, at least in principle, such a policy could succeed. But notice that $m$ is also negatively related to $\tau$. Moreover, as green preferences are not exogenously given and react to environmental quality through the broken
windows effect, we need to study the new dynamic system.\footnote{If we further assume that all resources raised by the tax ($\tau c_t$) are used to conservation activities, the problem immediately simplifies back to our baseline case. This happens because, substituting the budget into the environmental constraint, we have that $e_{t+1} =e_{t}-\gamma c_{t}+\sigma \left[ w-\textcolor{red}{(1+\tau)}c_t+\textcolor{blue}{\tau c_t}\right]$.}

\medskip

\begin{proposition}
    Introducing a consumption tax alters the first-order conditions of the intertemporal optimisation problem but does not modify the structure of the dynamic system in (\ref{M1_2D}). After substituting the tax-adjusted FOCs into the environmental constraint and the asynchronous updating of beliefs, the reduced-form map governing the evolution of ($e,\pi$) is identical to that of the baseline model.
\end{proposition}

\begin{proof}
    See Appendix A.3.
\end{proof}

\bigskip

The consumption tax reduces overall spending, but households try to compensate by also reducing their green investments. Consequently, the evolution of environmental conditions remains unchanged. The policy's ineffectiveness can be traced to its failure to address the underlying problem of social comparisons.




\section{Numerical simulations}

To provide a more concrete view of the properties of the system, we rely on numerical experiments. We begin by focusing on cases with a unique equilibrium solution but strong or increasing Veblen effects. Then, we assess those cases with multiple equilibria under a weak Veblenian mechanism. Since we are not calibrating a real economy, the results reported here are qualitative in nature and should be interpreted accordingly. We chose parameter values similar in magnitude to those in \href{#Antoci et al 2016}{Antoci et al. (2016)}, \href{#Caravaggio and Sodini 2023}{Caravaggio and Sodini (2023)}, and \href{#Davila-Fernandez et al 2025}{Davila-Fenandez et al. (2025)}. For completeness and reproducibility, we report below each figure the parameter set used to generate the corresponding scenario. Replication code is available at Mendeley Data.

\subsection{Instability under strong Veblen effects}

Our first experiment consists of studying the interaction between the responsiveness of green preferences to environmental quality, as captured by the intensity of choice parameter ($\beta$), and the extent to which agents care about the consumption of the reference group ($v$). Fig. \ref{beta_v} reports bifurcation diagrams showing that an increase in social comparisons is related to lower environmental quality and with agents valuing the environment less. For small values of $\beta$, the negative relationship is monotonic. Still, as we increase this parameter to 100 or 1000, a Neimark-Sacker bifurcation occurs, and we observe quasi-periodic oscillations in our variables of interest. The higher $\beta$, the lower $v$ has to be for the bifurcation to happen. Not only are Veblen effects associated with overconsumption and higher ecological footprints, but they may also lead to instability at very low levels of environmental quality.

\begin{figure}[tbp]
  \centering
  \includegraphics[width=2.7in]{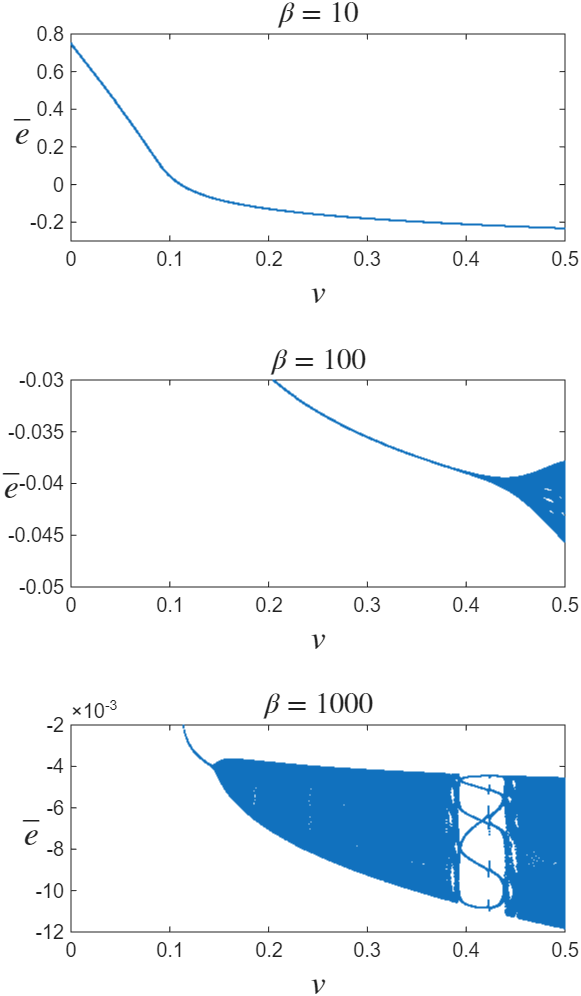}
  \hspace{1mm}
  \includegraphics[width=2.7in]{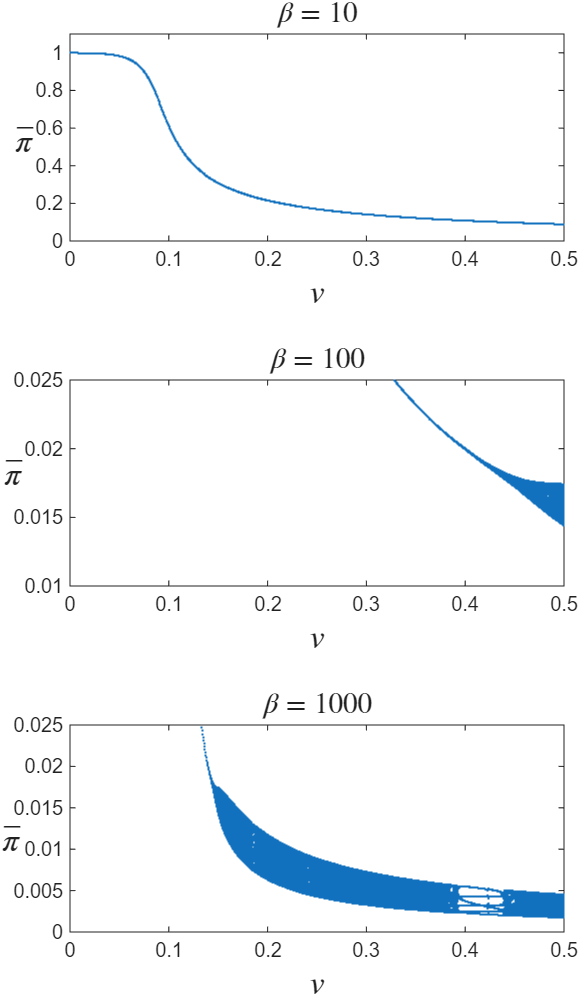}
  \caption{Social comparisons leading to environmental degradation and instability. Parameters $\alpha=0.75$, $\sigma=0.75$, $\gamma=1.5$, $w=1$, $\tilde{c}=3$, $\rho=0$. The bifurcation diagram starts with $\frac{v \tilde{c}}{w}<\frac{\sigma}{\sigma + \gamma}$, but the Neimark-Sacker occurs only when $\frac{v \tilde{c}}{w}>\frac{\sigma}{\sigma + \gamma}$.}\label{beta_v}
\end{figure}

The role of consumption driven by considerations of relative status as a source of instability can also be seen in Fig. \ref{alpha_v}. Blue and orange colours enable us to compare bifurcation diagrams for $\alpha=0.49$ and $\alpha=0.5$, respectively. An increase in the degree of inertia in green preferences brings stability to the system. In the limit, $\alpha=1$ implies preferences do not change, and our system becomes a 1-dimensional linear map. Reducing $\alpha$ not only anticipates the Neimark-Sacker bifurcation but also dramatically increases the amplitude of the oscillations. The blue plots suggest that a Flip bifurcation also occurs in this case. Overall, we confirm the negative relationship between $\bar{e}$, $\bar{\pi}$, and $v$. Endogenous fluctuations again appear to emerge only for very low or negative values of environmental quality, in a context where agents almost do not value it.

\begin{figure}[tbp]
  \centering
  \includegraphics[width=4.25in]{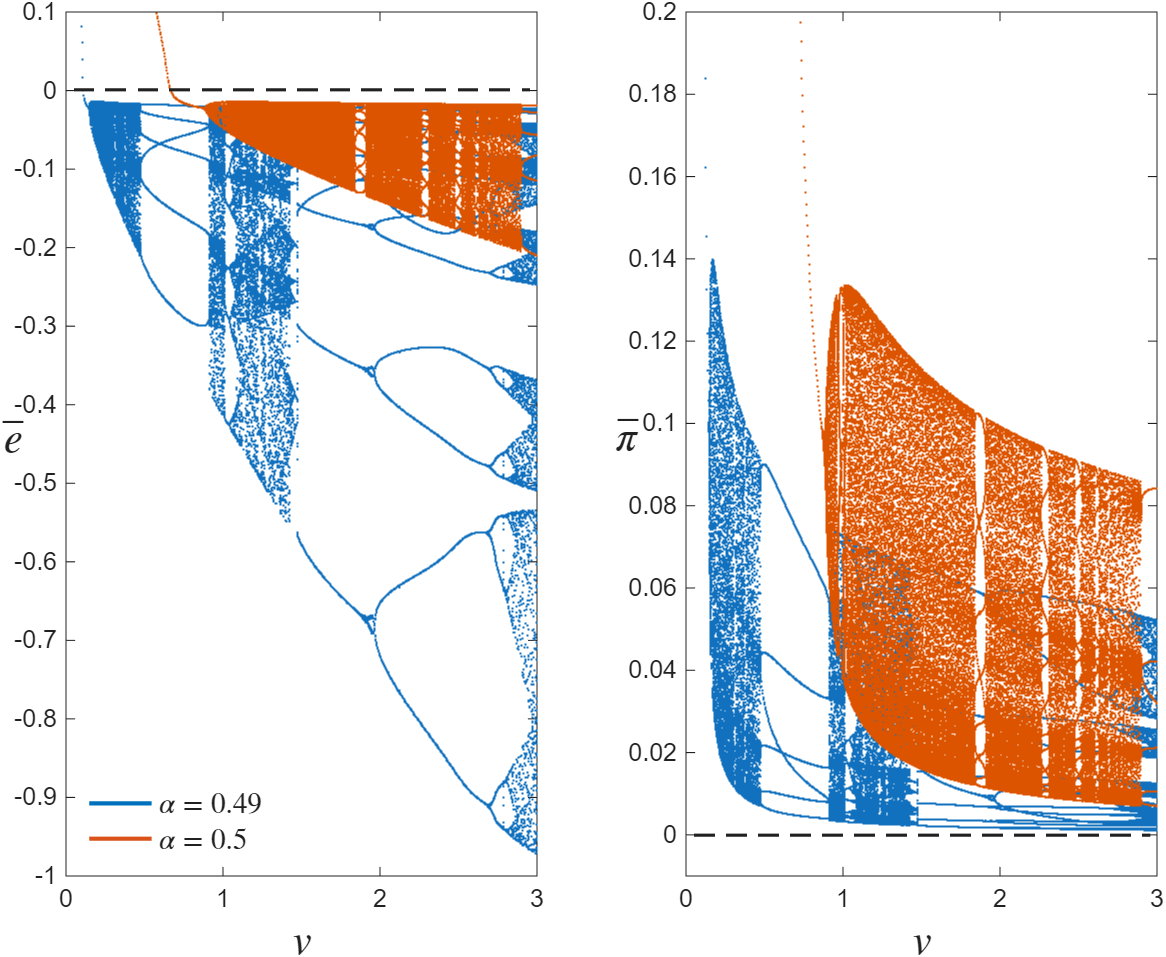}
  \caption{Social comparisons under different levels of green preferences inertia. Parameters $\beta=100$, $\sigma=0.75$, $\gamma=1.5$, $w=1$, $\tilde{c}=3$, $\rho=0$. The diagram starts with $\frac{v \tilde{c}}{w}<\frac{\sigma}{\sigma + \gamma}$, but bifurcations occur only when $\frac{v \tilde{c}}{w}>\frac{\sigma}{\sigma + \gamma}$.}\label{alpha_v}
\end{figure}

To confirm the stabilising role of preference inertia, Fig. \ref{alpha} reports the bifurcation diagram for $\alpha \in (0,1)$. Recall that for $\alpha=1$, green preferences do not change in our asynchronous updating of beliefs setup. We find that very low values of this parameter are associated with violent fluctuations. Under the current calibration, we need $\alpha \gtrapprox 0.6$ to regain stability. Notice that we are assuming Veblen effects are sufficiently high to guarantee that the bifurcation will eventually occur, despite the relatively small intensity of choice. That explains the extremely low values of environmental quality in this scenario.

\begin{figure}[tbp]
  \centering
  \includegraphics[width=4in]{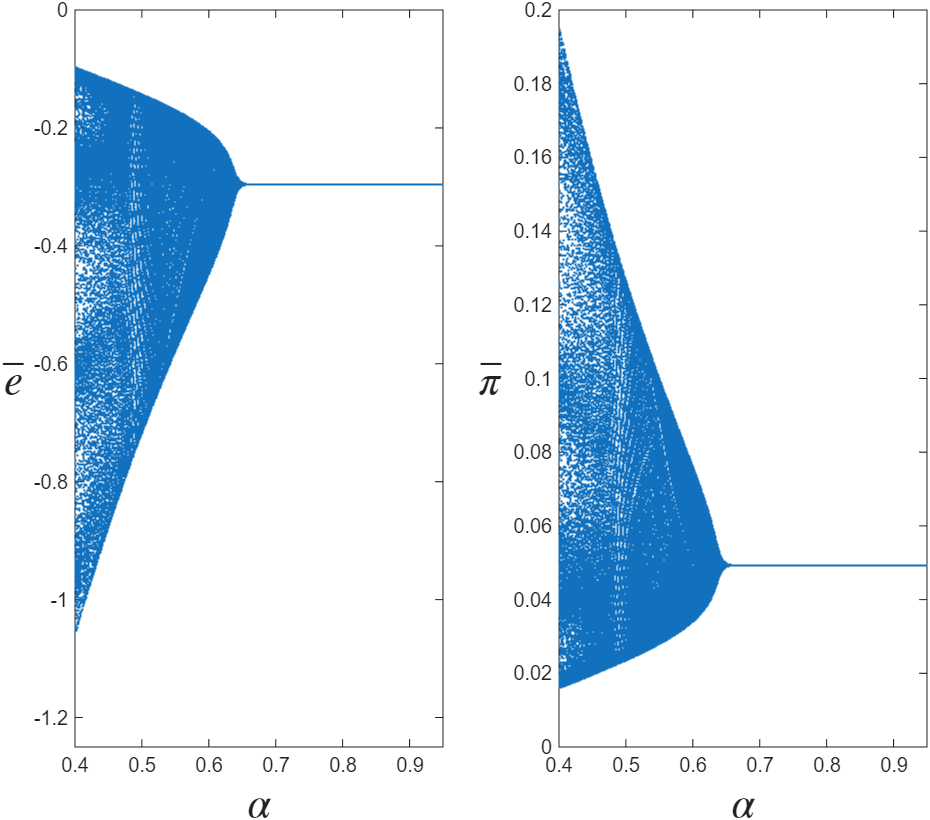}
  \caption{Instability under different levels of green preferences inertia. Parameters $\beta=10$, $v=1$, $\sigma=0.75$, $\gamma=1.5$, $w=1$, $\tilde{c}=3$, $\rho=0$. Notice that $\frac{v \tilde{c}}{w}>\frac{\sigma}{\sigma + \gamma}$ is satisfied $\forall \alpha$.}\label{alpha}
\end{figure}

\subsection{Materialistic trends under weak Veblen effects}

We are now ready to introduce the materialistic ``secular trend'' into the discussion. Our simulations so far have assumed $\rho=0$. Fig. \ref{rho} shows how this effect leads to the emergence and disappearance of multiple equilibria. To focus on this mechanism, we assume social comparisons are quite small. Panel (a), on the left, shows that when $\rho=0$, agents value the environment as much as consumption, given that $\bar{\pi}=1$. Panel (a), on the right, reports that setting $\rho=5$ leads us to the opposite situation. We still have a unique equilibrium, but agents do not value the environment at all. Increasing $\rho$ moves the S-shaped isocline along the horizontal axis, being responsible for the emergence and disappearance of multiple equilibria. In fact, for intermediate values of this parameter, a Pitchfork bifurcation occurs and two additional equilibria emerge, separated by a saddle. In the ``good'' solution, agents value the environment as much as consumption, while in the ``bad'' one, they get utility only from spending. Panel (b) plots the respective basins of attraction when $\rho=2.6$.

\begin{figure}[tbp]
  \centering
  \includegraphics[width=3.75in]{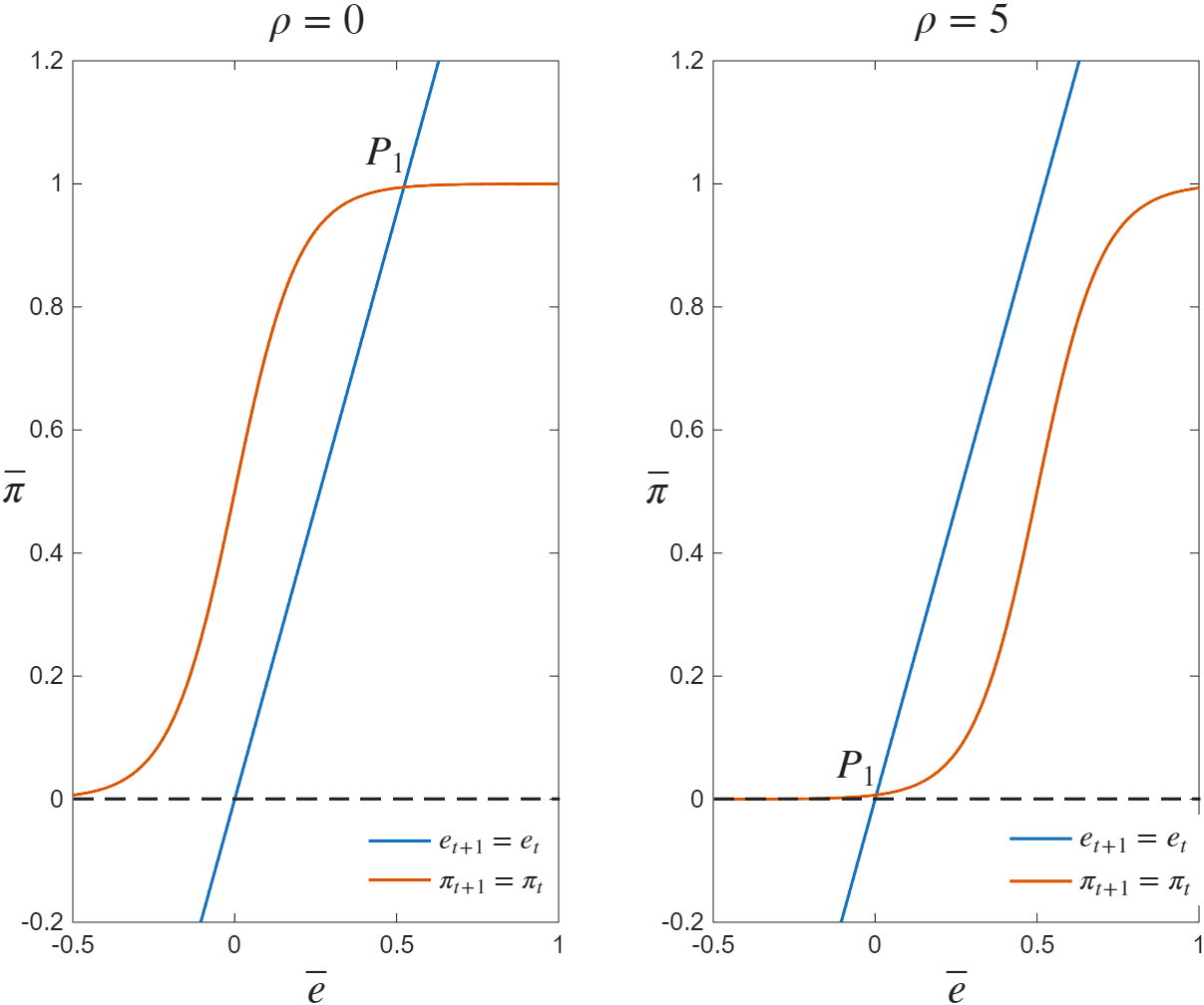}\vspace{10mm}
  \includegraphics[width=3.5in]{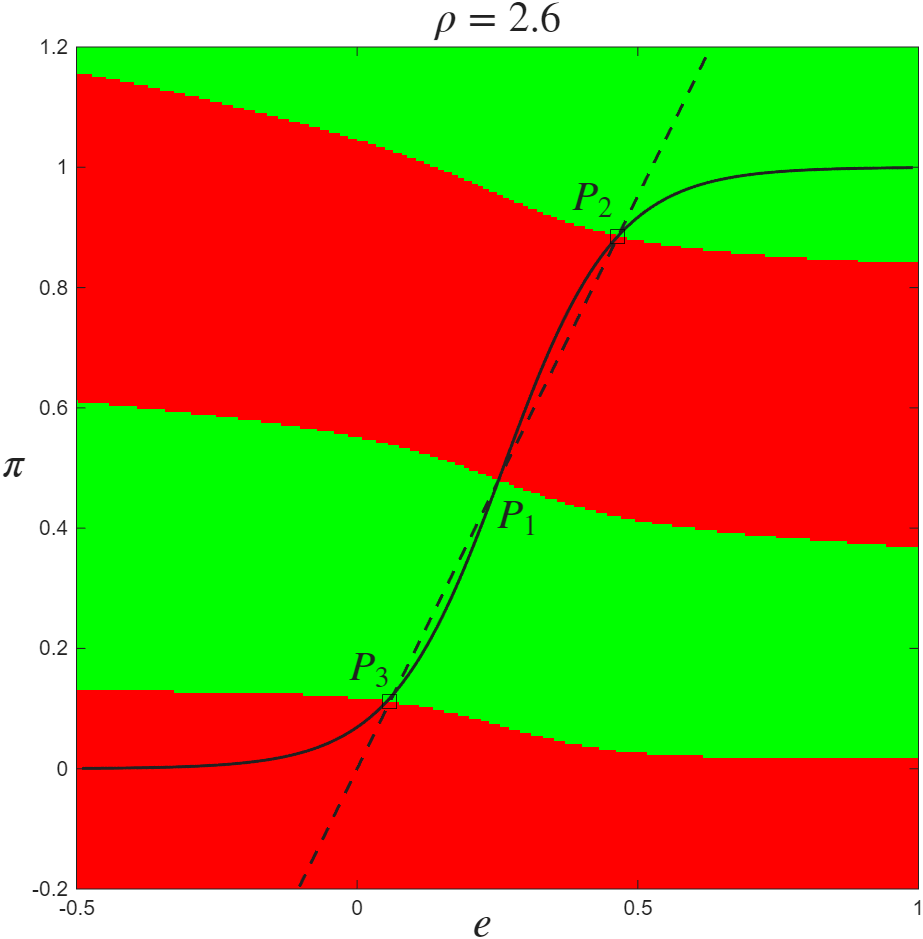}
  \caption{Multiple equilibria and the materialistic secular trend.}\label{rho}
    \vspace{0.2cm}
  \begin{minipage}{1\textwidth}
        \footnotesize
Notes: In panel (a), $P_1$ indicates the unique solution of the system and is determined by the intersection between blue and orange isoclines. In panel (b), the three equilibria correspond to the intersection of the continuous and dotted black isoclines. Green colours indicate convergence to the upper equilibrium $P_2$, while we mark in red all initial conditions leading to the lower equilibrium $P_3$. Parameters $\alpha=0.9$, $\sigma=0.75$, $\gamma=1.5$, $w=1$, $\tilde{c}=1$, $\beta=10$, $v=0.1$.
    \end{minipage}
\end{figure}

A striking feature of our basins of attraction is the presence of several disconnected regions. Such a result follows from our map being non-invertible. We mark in red those initial conditions leading to the ``bad '' equilibrium, while, in green, those leading to the ``good'' solution. Suppose an economy finds itself in the superior equilibrium. A small shock up or to the right will only temporarily perturb the system as it will converge to its initial state. However, any minor disturbance that reduces environmental quality or green preferences causes the system to jump to the inferior solution. This feature works in both ways. This means that once in a bad state, a small positive shock to preferences or environmental quality is enough to jump back to the ``good'' equilibrium. From an economic point of view, the alternation between green and red regions highlights the fragility of the economic-ecological relationship.

In our last set of experiments, we investigate the sensitivity of the basins of attraction to changes in the degree of inertia in preferences. Fig. \ref{rho_alpha} reports our main findings. Panel (a) sets $\alpha=0.5$. In panel (b), we increase this parameter to $\alpha=0.75$. Veblen effects are assumed sufficiently small so we can visualise the role of the  materialistic ``secular trend'' in generating multistability.\footnote{A multistable dynamical system is one that admits multiple stable equilibrium states or attractors, which may (or not) exhibit (a)periodic characteristics. Its long-run behaviour is highly sensitive to initial conditions, with trajectories converging to different stable states depending on their starting point (for a comprehensive review of theory and applications, see Pisarchik and Feudel, 2014). In economics, this feature can be understood as capturing path dependence. It also supports the idea that historical circumstances shape long-term outcomes.} The main distinguishing feature is that less inertia reduces the basins of attraction of the inferior equilibrium. This is an interesting result because it indicates that while less inertia increases instability under strong Veblen effects, it might bring stability to the ``good'' equilibrium when Veblen effects are moderate, even when agents are subject to an exogenous force that reduces how much they value the environment.

\begin{figure}[tbp]
  \centering
  \includegraphics[width=2.75in]{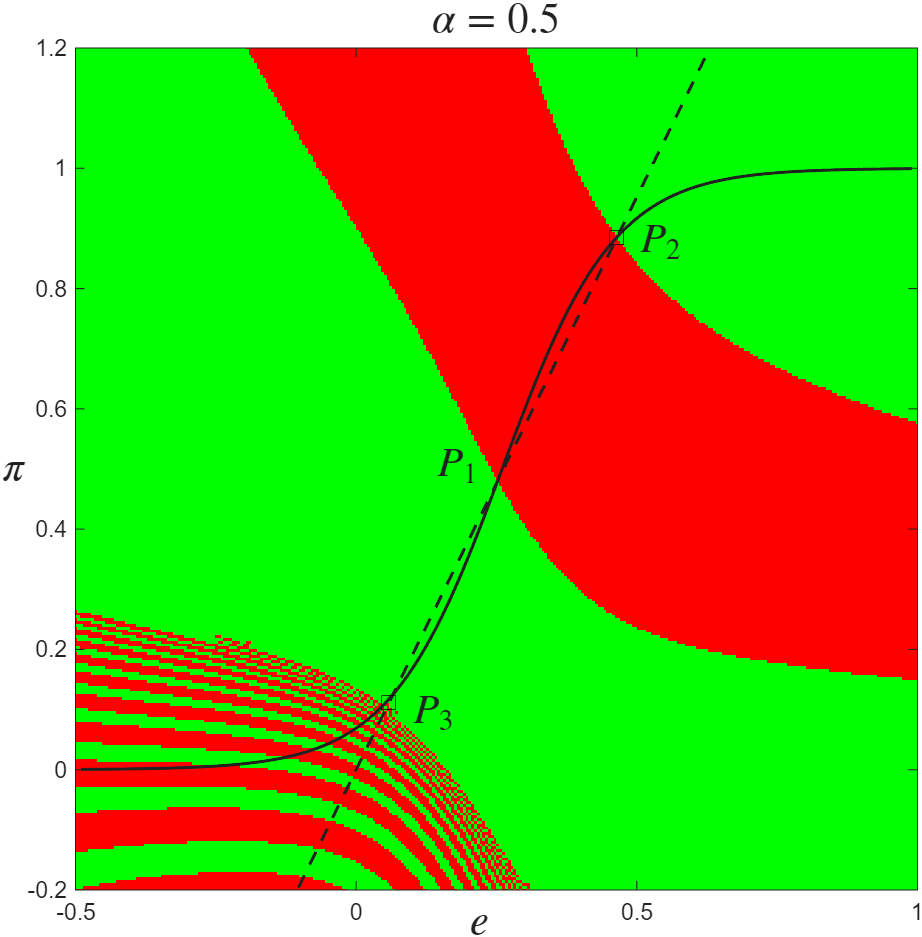}\hspace{10mm}
  \includegraphics[width=2.75in]{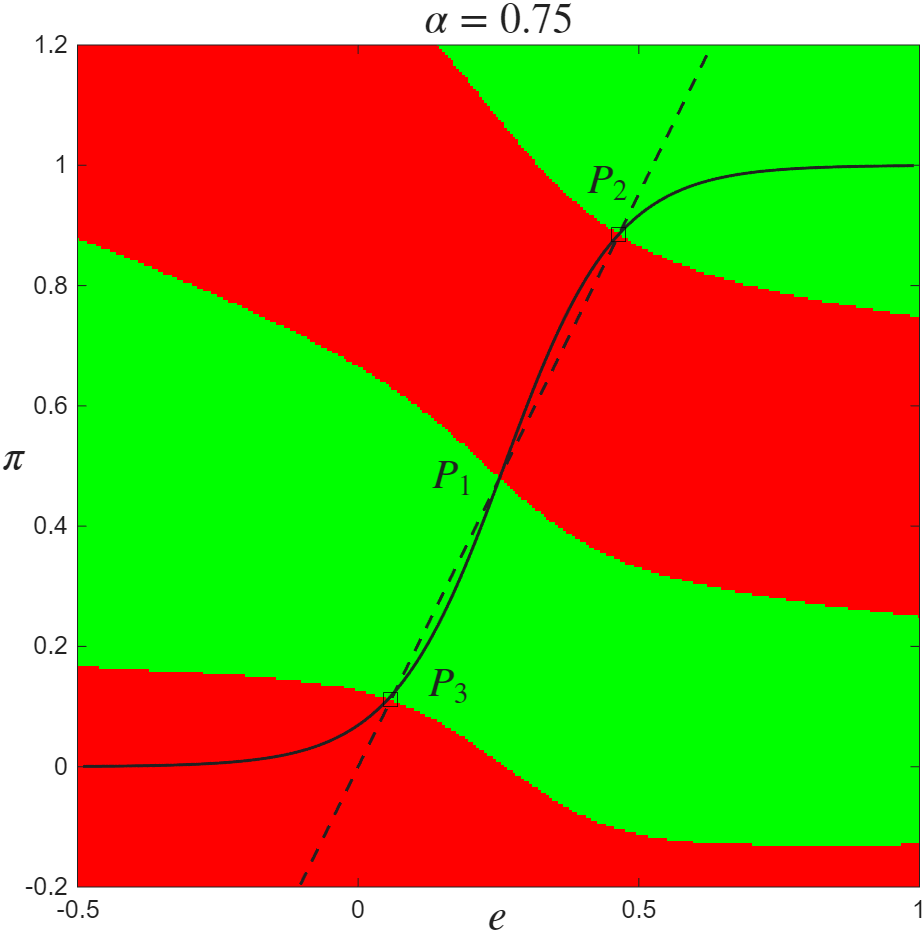}
  \caption{Basins of attraction.}\label{rho_alpha}
  \vspace{0.2cm}
  \begin{minipage}{1\textwidth}
        \footnotesize
Notes: Parameters $\rho=2.6$, $\sigma=0.75$, $\gamma=1.5$, $w=1$, $\tilde{c}=1$, $\beta=10$, $v=0.1$. Red colours indicate convergence to the lower equilibrium $P_3$ while, in green, we mark all initial conditions leading to the upper equilibrium $P_2$.
    \end{minipage}
\end{figure}


\section{Final considerations}

This paper developed an OLG model to study the interplay between Veblen effects, consumption and environmental damages. Building on two consolidated studies in both fields, we show that, along the optimal path, status-driven consumption leads to overconsumption. The latter is increasing with the reference group spending. For a given resource endowment, such relative status considerations imply lower investments in conservation activities, which is detrimental to the environment. The lower the quality of this public good, the less agents value it, reinforcing the increase in ecological footprints. This last effect was motivated by importing the concept of ``broken windows'' from the economic-criminology literature. We formalised it by innovatively applying discrete-choice theory with asynchronous belief updating.

When the Veblenian mechanism is weak, and we introduce a materialistic ``secular trend'' that captures ``the times we live in'' into the updating green preferences function, the model becomes compatible with multiple equilibria. A weak (strong) trend leads to a unique equilibrium solution in which agents (do not) value the environment. However, for intermediate values of this parameter, the model features two extreme stable solutions. The study of the basins of attraction in this case indicates the presence of several disconnected regions. Such a property can be interpreted as indicating the sensitivity and complexity of green preferences, as well as their interplay with other elements related to well-being.

We perform a set of numerical simulations that, though qualitative, allow us to get a better understanding of the main properties of the model. It was shown that the combination of strong Veblen and broken windows effects is associated with a Neimark-Sacker bifurcation. If both are sufficiently strong, green preferences fluctuate close to zero, so that agents value consumption only, while environmental quality is very low. This result is robust to variations in the degree of inertia in green preferences. When there is little persistence in agents' tastes, a Flip bifurcation might also occur. We do not claim that our findings are directly related to international experience or the trajectories of a particular country. Nonetheless, they suggest environmental vulnerability grows in parallel with status-driven consumption.

While highly stylised, our model offers insights that can be extended in several directions. For instance, we showed that taxing consumption may be counterproductive, as it does not directly address the underlying problem of social comparisons. The root of the issue lies in the high inequality that fuels them. However, we did not provide a formal account of the behaviour of the reference group(s). Moreover, the productive structure could be enriched by allowing for decisions on time allocation between work and leisure. In fact, \href{#Bowles and Park 2005}{Bowles and Park (2005)} showed that Veblen effects are associated with longer working hours. The negative consequences of environmental degradation on income, for example, through a damage function, also deserve further attention. These extensions would further join recent efforts to better understand the interplay between inequality, social comparisons, and environmental dynamics.

\newpage

\appendix

\numberwithin{equation}{section}
\numberwithin{figure}{section}

\section{Mathematical Appendix}

\subsection{Proof of Proposition 1}

We begin by demonstrating the existence of a unique equilibrium solution when Veblen effects are sufficiently strong:
\begin{equation*}
        \frac{v \tilde{c}}{w}>\frac{\sigma}{\sigma + \gamma}
    \end{equation*}
Notice that the first equilibrium condition of $M_1$ is given by:
\begin{equation*}
\bar{e}  =\bar{\pi}\left[  \sigma w-\left(  \sigma+\gamma\right)
\upsilon\tilde{c}\right]
\end{equation*}
Therefore, the corresponding isocline is linear, downward sloping, and crosses the origin. This case is illustrated by the blue line in Fig. \ref{strong_weak_veblen}, panel (a). The second isocline:
\begin{equation*}
\bar{\pi} =\frac{1}{1+\exp\left(  -\beta \bar{e}+ \rho  \right)
} \nonumber
\end{equation*}
is represented by the orange line. It is an upward sloping S-shaped curve that contains the point $(\bar{e},\bar{\pi})=(0,0.5)$ provided that $\rho=0$. As we increase this last parameter, the orange line moves to the right. Both isoclines can cross each other at only one point. Thus, the system admits a unique solution in such a scenario. Because the blue line necessarily crosses the origin and the orange line never crosses the horizontal axis, then $\bar{e}\leq 0$.

Moving on to the case in which Veblen effects are weak:
\begin{equation*}
        \frac{v \tilde{c}}{w}<\frac{\sigma}{\sigma + \gamma}
    \end{equation*}
Then, from the first isocline, $\bar{e}$ and $\bar{\pi}$ are positively related. This case is represented in Fig. \ref{strong_weak_veblen}, panel (b). Given that both isoclines have a positive slope, but the second is S-shaped, the model admits a unique equilibrium only when $\rho=0$. As we increase the materialistic trend, the orange curve moves to the right, and the model admits three equilibria, as shown in Fig. \ref{rho}. If we continue to increase this parameter, we recover a unique equilibrium.

\begin{figure}[h]
  \centering
  \includegraphics[width=2.75in]{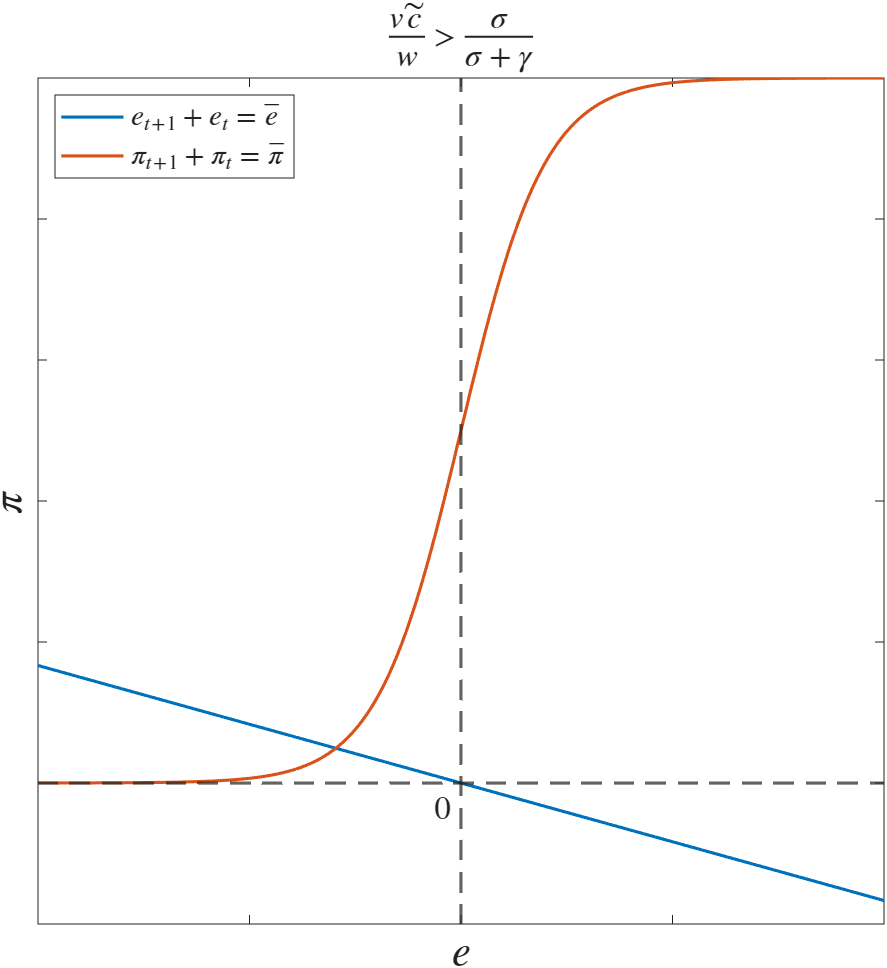} \hspace{10mm}
  \includegraphics[width=2.75in]{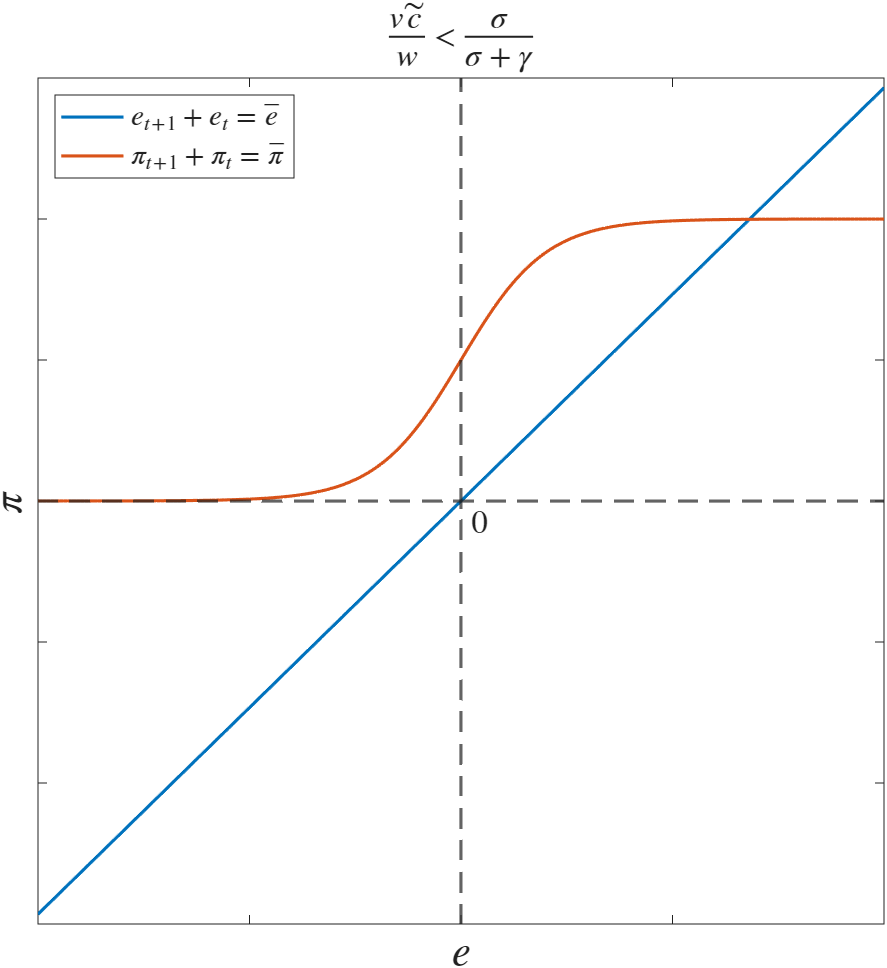}
  \caption{Strong vs weak Veblen effects and the possibility of multiple equilibria.}\label{strong_weak_veblen}
\end{figure}

\subsection{Proof of Proposition 2}

The Jacobian matrix of our 2-dimensional map evaluated at the neighborhood of a generic equilibrium point is given by:%
\[
J=%
\begin{bmatrix}
j_{11} & j_{12}\\
j_{21} & j_{22}%
\end{bmatrix}
\]
where%

\begin{align*}
j_{11} &  =\left.  \frac{\partial e_{_{t+1}}}{e_{_{t}}}\right\vert
_{e_{_{t+1}}=e_{_{t}}=\bar{e}}=\frac{\bar{\pi}}{1+\bar{\pi}}\\
j_{12} &  =\left.  \frac{\partial e_{_{t+1}}}{\pi_{t}}\right\vert _{e_{_{t+1}%
}=e_{_{t}}=\bar{e}}=\frac{\bar{e}+\sigma w-\left(  \sigma+\gamma\right)
v\tilde{c}}{\left(  1+\bar{\pi}\right)  ^{2}}%
\end{align*}%
\begin{align*}
j_{21} &  =\left.  \frac{\partial\pi_{_{t+1}}}{e_{_{t}}}\right\vert
_{\pi_{_{t+1}}=\pi_{_{t}}=\bar{\pi}}=\frac{\left(  1-\alpha\right)  \beta
\exp\left(  \rho-\beta\bar{e}\right)  }{\left[  1+\exp\left(  \rho-\beta
\bar{e}\right)  \right]  ^{2}}\\
j_{22} &  =\left.  \frac{\partial\pi_{_{t+1}}}{\pi_{t}}\right\vert
_{\pi_{_{t+1}}=\pi_{_{t}}=\bar{\pi}}=\alpha
\end{align*}
The coefficients of the characteristic equation are such that:%
\begin{align*}
tr\text{J}  & =j_{11}+j_{22}\\
& =\frac{\bar{\pi}}{1+\bar{\pi}}+\alpha>0
\end{align*}
and%
\begin{align*}
\det\text{J}  & =j_{11}j_{22}-j_{12}j_{21}\\
& =\frac{\alpha\bar{\pi}}{1+\bar{\pi}}-\left[  \frac{\bar{e}+\sigma w-\left(
\sigma+\gamma\right)  v\tilde{c}}{\left(  1+\bar{\pi}\right)  ^{2}}\right]
\frac{\left(  1-\alpha\right)  \beta\exp\left(  \rho-\beta\bar{e}\right)
}{\left[  1+\exp\left(  \rho-\beta\bar{e}\right)  \right]  ^{2}}\\
& =\frac{\alpha\bar{\pi}}{1+\bar{\pi}}-\left[  \frac{\bar{e}+\bar{e}/\bar{\pi
}}{\left(  1+\bar{\pi}\right)  ^{2}}\right]  \frac{\partial\pi_{_{t+1}}%
}{e_{_{t}}}\\
& =\frac{\alpha\bar{\pi}}{1+\bar{\pi}}-\frac{\bar{\pi}\eta+\eta}{\left(
1+\bar{\pi}\right)  ^{2}}\\
& =\frac{\alpha\bar{\pi}}{1+\bar{\pi}}-\frac{\eta}{1+\bar{\pi}}\\
& =\frac{\alpha\bar{\pi}-\eta}{1+\bar{\pi}}%
\end{align*}
where
\[
\eta\left(  \bar{e},\bar{\pi}\right)  =\frac{\bar{e}}{\bar{\pi}}\frac
{\partial\pi_{_{t+1}}}{e_{_{t}}}%
\]

For two-dimensional maps, the determination of the way in which the stability frontier is crossed is simplified by the use of stability conditions on the coefficients of the characteristic equation (\href{#Medio and Lines 2001}{Medio and Lines, 2001, p. 160}). The three
critical conditions are:%
\begin{align*}
\text{(i) }1-tr\text{J}+\det\text{J}  & >0\\
1-\frac{\bar{\pi}}{1+\bar{\pi}}-\alpha+\frac{\alpha\bar{\pi}-\eta}{1+\bar{\pi
}}  & >0\\
1-\alpha & >\eta
\end{align*}

\begin{align*}
\text{(ii) }1+trJ+\det J  & >0\\
1+\frac{\bar{\pi}}{1+\bar{\pi}}+\alpha+\frac{\alpha\bar{\pi}-\eta}{1+\bar{\pi
}}  & >0\\
1+\frac{\bar{\pi}}{1+\bar{\pi}}+\alpha\left(  1+\frac{\bar{\pi}}{1+\bar{\pi}%
}\right)    & >\frac{\eta}{1+\bar{\pi}}\\
\left(  1+\alpha\right)  \left(  1+\frac{\bar{\pi}}{1+\bar{\pi}}\right)    &
>\frac{\eta}{1+\bar{\pi}}\\
\left(  1+\alpha\right)  \left(  1+\bar{\pi}\right)  \left(  1+\frac{\bar{\pi
}}{1+\bar{\pi}}\right)    & >\eta\\
\left(  1+\alpha\right)  \left(  1+2\bar{\pi}\right)    & >\eta
\end{align*}
\begin{align*}
\text{(iii) }\det\text{J}  & <1\\
\frac{\alpha\bar{\pi}-\eta}{1+\bar{\pi}}  & <1\\
-\frac{1+\eta}{\bar{\pi}}  & <1-\alpha
\end{align*}
A violation of the first condition while the other two are satisfied results in a Fold bifurcation. If only the second condition is violated, a Flip bifurcation occurs. If conditions (i)-(ii) are satisfied while (iii) is violated, a Neimark-Sacker bifurcation occurs.

\subsection{Proof of Proposition 3}

Substituting the optimal trajectories of consumption and conservation
investments into the environment constraint, the evolution of $e$ is described
by:%
\begin{align*}
e_{t+1} &  =e_{t}-\gamma\left\{  \frac{\sigma w+e_{t}}{\left(  1+\pi
_{t}\right)  \left[  \gamma+\sigma\left(  1+\tau\right)  \right]  }+\left(
\frac{\pi_{t}}{1+\pi_{t}}\right)  v\tilde{c}\right\}  \\
&  +\sigma\left\{  \frac{\gamma w-\left(  1+\tau\right)  e_{t}}{\left(
1+\pi_{t}\right)  \left[  \gamma+\sigma\left(  1+\tau\right)  \right]
}+\left(  \frac{\pi_{t}}{1+\pi_{t}}\right)  \left(  w-v\tilde{c}\right)
\right\}  \\
&  =e_{t}-\frac{\gamma\left(  \sigma w+e_{t}\right)  }{\left(  1+\pi
_{t}\right)  \left[  \gamma+\sigma\left(  1+\tau\right)  \right]  }%
-\gamma\left(  \frac{\pi_{t}}{1+\pi_{t}}\right)  v\tilde{c}+\frac
{\sigma\left[  \gamma w-\left(  1+\tau\right)  e_{t}\right]  }{\left(
1+\pi_{t}\right)  \left[  \gamma+\sigma\left(  1+\tau\right)  \right]
}+\sigma\left(  \frac{\pi_{t}}{1+\pi_{t}}\right)  \left(  w-v\tilde{c}\right)
\\
&  =\frac{e_{t}\left(  1+\pi_{t}\right)  \left[  \gamma+\sigma\left(
1+\tau\right)  \right]  -\gamma\left(  \sigma w+e_{t}\right)  }{\left(
1+\pi_{t}\right)  \left[  \gamma+\sigma\left(  1+\tau\right)  \right]  }\\
&  \frac{-\gamma\pi_{t}\left[  \gamma+\sigma\left(  1+\tau\right)  \right]
v\tilde{c}+\sigma\left[  \gamma w-\left(  1+\tau\right)  e_{t}\right]
+\sigma\pi_{t}\left[  \gamma+\sigma\left(  1+\tau\right)  \right]  \left(
w-v\tilde{c}\right)  }{\left(  1+\pi_{t}\right)  \left[  \gamma+\sigma\left(
1+\tau\right)  \right]  }\\
&  =\frac{e_{t}\left(  1+\pi_{t}\right)  \left[  \gamma+\sigma\left(
1+\tau\right)  \right]  -\gamma\sigma w-\gamma e_{t}}{\left(  1+\pi
_{t}\right)  \left[  \gamma+\sigma\left(  1+\tau\right)  \right]  }\\
&  \frac{-\gamma\pi_{t}\left[  \gamma+\sigma\left(  1+\tau\right)  \right]
v\tilde{c}+\sigma\gamma w-\sigma\left(  1+\tau\right)  e_{t}+\sigma\pi
_{t}\left[  \gamma+\sigma\left(  1+\tau\right)  \right]  \left(  w-v\tilde
{c}\right)  }{\left(  1+\pi_{t}\right)  \left[  \gamma+\sigma\left(
1+\tau\right)  \right]  }\\
&  =\frac{e_{t}\left(  1+\pi_{t}\right)  \left[  \gamma+\sigma\left(
1+\tau\right)  \right]  -\gamma e_{t}-\sigma\left(  1+\tau\right)  e_{t}%
}{\left(  1+\pi_{t}\right)  \left[  \gamma+\sigma\left(  1+\tau\right)
\right]  }\\
&  \frac{-\gamma\pi_{t}\left[  \gamma+\sigma\left(  1+\tau\right)  \right]
v\tilde{c}+\sigma\pi_{t}\left[  \gamma+\sigma\left(  1+\tau\right)  \right]
w-\sigma\pi_{t}\left[  \gamma+\sigma\left(  1+\tau\right)  \right]  v\tilde
{c}}{\left(  1+\pi_{t}\right)  \left[  \gamma+\sigma\left(  1+\tau\right)
\right]  }%
\end{align*}%
\begin{align*}
&  =\frac{\left\{  \left(  1+\pi_{t}\right)  \left[  \gamma+\sigma\left(
1+\tau\right)  \right]  -\gamma-\sigma\left(  1+\tau\right)  \right\}  e_{t}%
}{\left(  1+\pi_{t}\right)  \left[  \gamma+\sigma\left(  1+\tau\right)
\right]  }\\
&  +\frac{\sigma\pi_{t}\left[  \gamma+\sigma\left(  1+\tau\right)  \right]
w}{\left(  1+\pi_{t}\right)  \left[  \gamma+\sigma\left(  1+\tau\right)
\right]  }-\frac{\left\{  \gamma\pi_{t}\left[  \gamma+\sigma\left(
1+\tau\right)  \right]  +\sigma\pi_{t}\left[  \gamma+\sigma\left(
1+\tau\right)  \right]  \right\}  v\tilde{c}}{\left(  1+\pi_{t}\right)
\left[  \gamma+\sigma\left(  1+\tau\right)  \right]  }\\
&  =\frac{\left\{  \gamma+\sigma\left(  1+\tau\right)  +\pi_{t}\left[
\gamma+\sigma\left(  1+\tau\right)  \right]  -\gamma-\sigma\left(
1+\tau\right)  \right\}  e_{t}}{\left(  1+\pi_{t}\right)  \left[
\gamma+\sigma\left(  1+\tau\right)  \right]  }\\
&  +\left(  \frac{\pi_{t}}{1+\pi_{t}}\right)  \sigma w-\left(  \frac{\pi_{t}%
}{1+\pi_{t}}\right)  \frac{\left(  \gamma+\sigma\right)  \left[  \gamma
+\sigma\left(  1+\tau\right)  \right]  v\tilde{c}}{\gamma+\sigma\left(
1+\tau\right)  }\\
&  =\left(  \frac{\pi_{t}}{1+\pi_{t}}\right)  \frac{\left[  \gamma
+\sigma\left(  1+\tau\right)  \right]  e_{t}}{\left[  \gamma+\sigma\left(
1+\tau\right)  \right]  }+\left(  \frac{\pi_{t}}{1+\pi_{t}}\right)  \sigma
w-\left(  \frac{\pi_{t}}{1+\pi_{t}}\right)  \left(  \gamma+\sigma\right)
v\tilde{c}\\
&  =\left(  \frac{\pi_{t}}{1+\pi_{t}}\right)  e_{t}+\left(  \frac{\pi_{t}%
}{1+\pi_{t}}\right)  \sigma w-\left(  \frac{\pi_{t}}{1+\pi_{t}}\right)
\left(  \gamma+\sigma\right)  v\tilde{c}\\
&  =\left(  \frac{\pi_{t}}{1+\pi_{t}}\right)  \left[  e_{t}+\sigma w-\left(
\gamma+\sigma\right)  v\tilde{c}\right]
\end{align*}
Therefore, the structure of the first difference equation of our map is the
same. From the asyncronous updating of beliefs, it is also clear
that the final expression does not change as it does not depend on consumption decisions.


\newpage

\begin{thebibliography}{99}

\bibitem{}\hypertarget{Annicchiarico et al 2024}{Annicchiarico, B., Di Dio, F., Diluiso, F. (2024). Climate actions, market beliefs, and monetary policy. \emph{Journal of Economic Behavior \& Organization} 218: 176-208.}

\bibitem{}\hypertarget{Antoci et al 2016}{Antoci, A., Gori, L., Sodini, M. (2016). Nonlinear dynamics and global indeterminacy in an overlapping generations model with environmental resources. \emph{Communications in Nonlinear Science and Numerical Simulation} 38: 59-71.}

\bibitem{}\hypertarget{Antoci et al 2019}{Antoci, A., Gori, L., Sodini, M., Ticci, E. (2019). Maladaptation and global indeterminacy. \emph{Environment and Development Economics} 24(6): 643-659.}

\bibitem{}\hypertarget{Bagwell and Bernheim 1996}{Bagwell, L., Bernheim, D. (1996). Veblen effects in a theory of conspicuous consumption. \emph{American Economic Review} 86(3): 349-373.}

\bibitem{}\hypertarget{Behringer et al 2024}{Behringer, J., Gonzalez-Granda, M., van Treeck, T. (2024). Varieties of the rat race: Working hours in the age of abundance. \emph{Socio-Economic Review} 22(1): 141-168.}

\bibitem{}\hypertarget{Bowles and Park 2005}{Bowles, S., Park, Y. (2005). Emulation, inequality, and work hours: Was Thorsten Veblen right? \emph{Economic Journal} 115(507): F397-F412.}

\bibitem{}\hypertarget{Brock and Durlauf 2001}{Brock, W., Durlauf, S. (2001). Discrete choice with social interactions. \emph{Review of Economic Studies} 68: 235-260.}

\bibitem{}\hypertarget{Brock and Hommes 1998}{Brock, W., Hommes, C. (1998). Heterogeneous beliefs and routes to chaos in a simple asset pricing model. \emph{Journal of Economic Dynamics \& Control} 22: 1235-1274.}

\bibitem{}\hypertarget{Cafferata et al 2021}{Cafferata , A., Davila-Fernandez, M., Sordi, S. (2021). Seeing what can (not) be seen: Confirmation bias, employment dynamics and climate change. \emph{Journal of Economic Behavior \& Organization} 189, 567-586.}

\bibitem{}\hypertarget{Campiglio et al 2024}{Campiglio, E., Lamperti, F., Terranova, R. (2024). Believe me when I say green! Heterogeneous expectations and climate policy uncertainty \emph{Journal of Economic Dynamics \& Control} 165, 104900.}

\bibitem{}\hypertarget{Caravaggio and Sodini 2023}{Caravaggio, A., Sodini, M. (2023). Environmental sustainability, nonlinear dynamics and chaos reloaded: 0 matters! \emph{Communications in Nonlinear Science and Numerical Simulation} 117, 106908.}

\bibitem{}\hypertarget{Corman and Mocan 2002}{Corman, H., Mocan, N. (2002). Carrots, sticks and broken windows. \emph{NBER WP Series} 9061, 1-54.}

\bibitem{}\hypertarget{Davila-Fernandez et al 2025}{Davila-Fernandez, M., Giombini, G., Sanchez-Carrera, E. (2025). Climateflation and monetary policy in an environmental OLG growth model. \emph{Environmental \& Resource Economics}. Online First. DOI: 10.1007/s10640-025-00991-1.}

\bibitem{}\hypertarget{de la Croix and Gosseries 2012}{de la Croix, D., Gosseries, A. (2012). The natalist bias of pollution control. \emph{Journal of Environmental Economics and Management} 63, 271–287.}

\bibitem{}\hypertarget{Franke and Westerhoff 2017}{Franke, R., Westerhoff, F. (2017). Taking stock: A rigorous modelling of animal spirits in macroeconomics. \emph{Journal of Economic Surveys} 31(5), 1152-1182.}

\bibitem{}\hypertarget{Grassetti et al 2025}{Grassetti, F., Sanchez-Carrera, E., Seegmuller, T. (2025). The role of government intervention in balancing economic growth and environmental sustainability: A global dynamics approach across economies. \emph{Macroeconomic Dynamics} 29, 1-26.}

\bibitem{}\hypertarget{Hommes et al 2005}{Hommes, C., Huang, H., Wang, D. (2005). A robust rational route to randomness in a simple asset pricing model. \emph{Journal of Economic Dynamics \& Control} 29(6): 1043-1072.}

\bibitem{}\hypertarget{Jaimes 2023}{Jaimes, R. (2023). Optimal climate and fiscal policy in an OLG economy. \emph{Journal of Public Economic Theory} 25(4): 727-752.}

\bibitem{}\hypertarget{John and Pecchenino 1994}{John, A., Pecchenino, R. (1994). An overlapping generations model of growth and the environment. \emph{Economic Journal} 104(427): 1393-1410.}

\bibitem{}\hypertarget{Mariani et al 2010}{Mariani, F., Pérez-Barahona, A., Raffin, N. (2010). Life expectancy and the environment. \emph{Journal of Economic Dynamics \& Control} 34(4): 798-815.}

\bibitem{}\hypertarget{Medio and Lines 2001}{Medio, A., Lines, M. (2001). \emph{Nonlinear Dynamics: A Primer}. Cambridge University Press, Cambridge.}

\bibitem{}\hypertarget{Miceli and Segerson 2024}{Miceli, T., Segerson, K. (2024). The broken-windows theory of crime: A Bayesian approach. \emph{International Review of Law and Economics} 80: 106233.}

\bibitem{}\hypertarget{Niinimaki et al 2020}{Niinimäki, K., Peters, G., Dahlbo, H., Perry, P., Rissanen, T., Gwilt, A. (2020). The environmental price of fast fashion. \emph{Nature Reviews Earth \& Environment} 1: 189-200.}

\bibitem{}\hypertarget{Pisarchik and Feudel 2014}{Pisarchik, A., Feudel, U.  (2014). Control of multistability. \emph{Physics Reports} 540(4): 167-218.}

\bibitem{}\hypertarget{Schongart et al 2025}{Schöngart, S., Nicholls, Z., Hoffmann, R., Pelz, S., Schleussner, C.-F. (2025). High-income groups disproportionately contribute to climate extremes worldwide. \emph{Nature Climate Change}, 15, 627--633. DOI: 10.1038/s41558-025-02325-x.}

\bibitem{}\hypertarget{van der Weele et al 2021}{van der Weele, J., Flynn, M., van der Wolk, R. (2021). Broken windows effect.  In: Marciano, A., Ramello, G. (Eds.) \emph{Encyclopedia of Law and Economics}. Springer, New York.}

\bibitem{}\hypertarget{Veblen 1899}{Veblen, T. (1899)[1992]. \emph{The Theory of the Leisure Class}. Routledge, New York.}

\bibitem{}\hypertarget{Wei and Aadland 2022}{Wei, S., Aadland, D. (2022). Physical capital, human capital, and the health effects of pollution in an OLG model. \emph{Macroeconomic Dynamics} 26, 1522-1563.}

\bibitem{}\hypertarget{Wilson and Kelling 1982}{Wilson, J., Kelling, G. (1982). Broken windows. \emph{The Atlantic} 249(3), 29-38.}

\bibitem{}\hypertarget{Zeppini 2015}{Zeppini, P. (2015). A discrete choice model of transitions to sustainable technologies. \emph{Journal of Economic Behavior \& Organization} 112, 187-203.}

\bibitem{}\hypertarget{Zhang 1999}{Zhang, X. (1999). Environmental sustainability, nonlinear dynamics and chaos. \emph{Economic Theory} 14, 489-500.}

\end{thebibliography}
\end{document}